\documentclass[10pt]{article}
\usepackage[utf8]{inputenc}
\usepackage{amsmath,amssymb,bbm}
\newcommand{\argmax}{\mathop{\arg\max}}
\newcommand{\argmin}{\mathop{\arg\min}}
\newcommand{\bs}{\boldsymbol}

\usepackage[table]{xcolor}
\definecolor{Gray}{gray}{0.9}

\usepackage{algorithm}
\usepackage{algorithmicx}
\usepackage{algpseudocode}
\usepackage{booktabs}

\usepackage[whole]{bxcjkjatype}

\usepackage{graphicx}[dvipdfmx]

\usepackage{enumerate}
\usepackage{subcaption}
\usepackage{mdframed}
\usepackage{natbib}
\usepackage{hyperref}
\hypersetup{
    colorlinks = true,
	citecolor=blue,
	linkcolor=blue
}
\usepackage{appendix}
\usepackage{colortbl}

\usepackage{geometry}
\allowdisplaybreaks[4]

\usepackage{authblk}
\usepackage{enumerate}
\usepackage{amsthm}
\theoremstyle{definition}

\newtheorem*{theorem*}{Theorem}

\newtheorem*{note*}{Note}

\newtheorem{ex}{Example}

\newtheorem{prop}{Proposition}

\newcommand{\setX}{\mathcal{X}}
\newcommand{\setZ}{\mathcal{Z}}
\newcommand{\setA}{\mathcal{A}}
\newcommand{\clustering}{\mathfrak{C}}
\newcommand{\normaldist}{\mathcal{N}}

\usepackage{cprotect}
\usepackage{booktabs}
\usepackage{multirow}

\usepackage{fourier}

\title{A Greedy and Optimistic Approach to Clustering \\ with a Specified Uncertainty of Covariates}
\author[1,2]{Akifumi Okuno}
\author[3,1,4]{Kohei Hattori}
\affil[1]{The Institute of Statistical Mathematics}
\affil[2]{RIKEN Center for Advanced Intelligence Project}
\affil[3]{National Astronomical Observatory of Japan}
\affil[4]{Department of Astronomy, University of Michigan}

\date{\empty}

\begin{document}

\maketitle

\begin{abstract}
In this study, we examine a clustering problem in which the covariates of each individual element in a dataset are associated with an uncertainty specific to that element.
More specifically, we consider a clustering approach in which a pre-processing applying a non-linear transformation to the covariates is used to capture the hidden data structure. To this end, we approximate the sets representing the propagated uncertainty for the pre-processed features empirically. 
To exploit the empirical uncertainty sets, we propose a greedy and optimistic clustering (GOC) algorithm that finds better feature candidates over such sets, yielding more condensed clusters. 
As an important application, we apply the GOC algorithm to synthetic datasets of the orbital properties of stars generated through our numerical simulation mimicking the formation process of the Milky Way. 
The GOC algorithm demonstrates an improved performance in finding sibling stars originating from the same dwarf galaxy. 
These realistic datasets have also been made publicly available. 
\end{abstract}


\section{Introduction}

The discovery of distinct groups of unlabeled individuals using their covariates, a process called \emph{clustering}, has been a fundamental statistical problem in the fields of psychology~\citep{borgen1987applying,henry2005cluster}, astronomy~\citep{roederer2018kinematics,Helmi2020ARA&A..58..205H, Yuan2020ApJ...891...39Y}, and biology~\citep{ben1999clustering,nugent2010overview}, among other areas. 
Owing to its versatility and numerous applications, many types of clustering algorithms have been developed, including $K$-means~\citep{macqueen1967some}, 
% $K$-means++~\citep{arthur2007k}, 
mean-shift~\citep{fukunaga1975estimation,cheng1995mean}, 
% DBSCAN \citep{ester1996density}, 
% OPTICS~\citep{ankerst1999optics}, 
spectral clustering~\citep{chung1997spectral,von2007tutorial}, 
convex clustering ~\citep{pelckmans2005convex}, 
and likelihood-based approaches using stochastic block models ~\citep{holland1983stochastic} and Gaussian mixture models ~\citep{mclachlan2000finite}. 
% and Dirichlet process mixture models ~\citep{neal2000markov}. 
For comprehensive surveys of various clustering algorithms, see \citet{jain1988algorithms}, \citet{everitt1993cluster}, and \citet{xu2005survey}. 

Although the aforementioned clustering methods assume that an instance is observed for each individual covariate, in practice, a covariate may have uncertainty caused by limited data observability and noisy measurements. 
A set representing such uncertainty is called an \emph{uncertainty set}~\citep{ben2002robust,bertsimas2011theory}. 
% (or an ambiguity set~\citep{wiesemann2014distributionally,mohajerin2018data}). 
Robust optimization (RO)~\citep{ben2002robust,bertsimas2011theory} is an approach used to exploit an uncertainty set. RO minimizes the worst-case loss functions (over the uncertainty sets for the covariates of all individuals), with application to statistical problems including classification~ \citep{xu2009robustness,takeda2013unified} and clustering~ \citep{vo2016robust}. 
In contrast to the pessimistic approach of RO, several studies have reported that an optimistic attitude, that is, optimizing the best case (instead of the worst case), demonstrates an improved performance for various problems, including classification~\citep{jinbo2004support}, the multi-armed bandit problem~\citep{bubeck2012regret}, and Bayesian optimization~\citep{srinivas2010gaussian,nguyen2019optimistic,nguyen2019calculating}.

For computational tractability, both the pessimistic and optimistic approaches described above employ a convex uncertainty set for each individual covariate, typically a small ball equipped with the $p$-norm ($p = 1,2$) centered at the instance of an observed covariate ~\citep{ben2002robust,bertsimas2011theory,vo2016robust}. 
Although existing studies have mainly assumed that an uncertainty set is simply a fixed-sized ball around a covariate instance, in practical situations, such uncertainty sets can be specified by background knowledge accumulated over time. We assume here that the uncertainty sets are specified by the users, allowing several entries in some instances to have greater uncertainty than others.

In this study, we consider a clustering problem with covariates whose uncertainty sets $\setZ_1,\ldots,\setZ_n$ for each individual $i=1,2,\ldots,n$ are user-specified. 
More practically, we consider a situation in which the covariate $Z_i$ is further pre-processed by applying a nonlinear function $f$ (prior to the clustering analysis) used to capture the latent data structure through the pre-processed feature $X_i = f(Z_i)$. 
This pre-processing step $f$ is also expected to remove redundant information harmful to the clustering process. 
Because the explicit form of an uncertainty set for a pre-processed feature is difficult to obtain, we first generate empirical uncertainty sets  
% $\tilde{\setX}_1^{(m_1)},\tilde{\setX}_2^{(m_2)},\ldots,\tilde{\setX}_n^{(m_n)}$ 
that approximate the underlying feature uncertainty sets $\setX_1,\setX_2,\ldots,\setX_n$. To fully exploit them, we propose a simple greedy and optimistic clustering (GOC) algorithm, which greedily seeks feature candidates over the sets % $\chi_1 \in \tilde{\setX}_1^{(m_1)},\chi_2 \in \tilde{\setX}_2^{(m_2)},\ldots,\chi_n \in \tilde{\setX}_n^{(m_n)}$ 
that yield more condensed clusters. 
The proposed GOC algorithm simply iterates the following steps: GOC algorithm (i) computes temporal cluster assignments of the current feature candidates using an \emph{arbitrary} clustering oracle (e.g., $K$-means or other clustering method listed above), and (ii) updates the feature candidates (and their temporal cluster assignments simultaneously) to reduce each cluster radius. 

\vspace{-1em}
\begin{figure}[!ht]
\centering
\begin{minipage}{0.45\textwidth}
\centering
\includegraphics[width=0.75\textwidth]{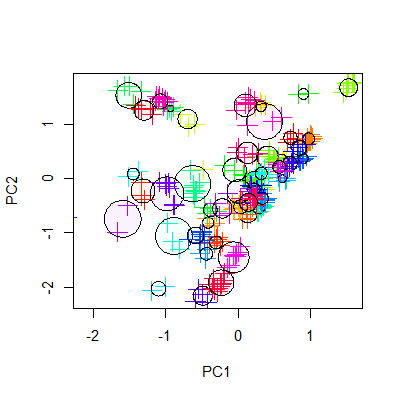}
\subcaption{Conventional $K$-means ($K=20$).}
\label{fig:conventional}
\end{minipage}
\hspace{1em}
\begin{minipage}{0.45\textwidth}
\centering
\includegraphics[width=0.75\textwidth]{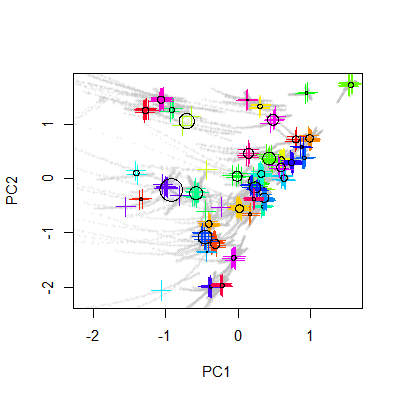}
\subcaption{GOC equipped with $K$-means ($K=20$).}
\label{fig:optimistic}
\end{minipage}
\caption{Principal components of the clusters obtained: 
(\subref{fig:conventional}) $K$-means applied to the mean vectors of each uncertainty set, and the 
(\subref{fig:optimistic}) GOC algorithm with $K$-means applied to the uncertainty sets (shown in gray). }
\label{fig:illustration_intro}
\end{figure}

As an important application of the GOC algorithm, we consider a clustering of stars in the Milky Way to find groups of stars ({\it sibling stars}) with similar orbital properties (so-called {\it orbital actions}).  
We generate synthetic datasets 
% (orbital actions and their uncertainty sets) 
of such stars by simulating the formation of the Milky Way. 
% In generating these datasets, we transform the observed position and velocity $\zeta_i \in \mathbb{R}^6$ of each star $i$ (which can be regarded as an instance of the covariate $Z_i$) into its orbital action using a nonlinear function $f$, and the non-negligible observational uncertainty then propagates toward the transformed quantity. 
See Example~\ref{ex:stars} in Section \ref{sec:problem_setting} for further details. 
We apply the GOC algorithm to these realistic datasets to find sibling stars, as illustrated in Figure~\ref{fig:illustration_intro}.
In comparison to conventional $K$-means applied to the mean vectors of each uncertainty set, the GOC algorithm using $K$-means yields a more condensed clusters of stars. The GOC algorithm also improves the clustering scores.
These datasets have also been made publicly available in our repository~(\url{https://github.com/oknakfm/GOC}). 

Although in this study the GOC algorithm is evaluated by leveraging realistic datasets whose true cluster assignments are known, the GOC algorithm has been applied to a real-world orbital action dataset in another study of ours~\citep{Hattori2022}, the results of which will be submitted to an astronomy journal.

\subsection{Related Works} 
A similar approach can be found in \citet{Ngai2006efficient}, which assumes that the uncertainty set $\setX_i$ consists of finite points and considers the minimum box $B_i (\supset \setX_i)$; (a bound of) the Hausdorff distance between the boundary $\partial B_i$ and cluster center is used for $K$-means instead of the Euclidean distance therein. 
However, \citet{Ngai2006efficient} does not perfectly fit our setting as it implicitly assumes the convexity of the set $\setX_i$ (also see Section~\ref{subsec:discussion2_convex} for discussion). 
%(Possible-world): 
Another similar approach is possible-world~(PW) model, which considers all the possible combinations of feature candidates (called ``worlds'') $\Xi_1,\Xi_2,\ldots, \Xi_{|\setA_n|}\in \setA_n:=\setX_1 \times \setX_2 \times \cdots \setX_n$, applies a clustering algorithm for each world $\Xi_1,\Xi_2,\ldots$ (in parallel), and aggregates all the clustering results. See, e.g., \citet{Volk2009clustering}, \citet{Zufle2014representative} and \cite{Liu2021RPC}. 
In contrast to our approach finding only the optimistic candidates, PW models overall require much more computational cost as the number of (even a subset of) possible worlds $\setA_n$ is numerous. 

%(Expected-dist. UK), 
Another direction for clustering uncertain data employs the probability density function $p_i$ of the feature $X_i$ ($i=1,2,\ldots,n$). 
While the simple $K$-means~\citep{macqueen1967some} minimizes the squared Euclidean distance between the feature instance and the cluster centers, UK-means~\citep{Chau2006uncertain} considers the expectation of the distance between feature and the cluster centers (with respect to $(p_1,p_2,\ldots,p_n)$). As pointed out in \citet{Lee2006reducing} and \citet{Cormode2008approximation}, UK-means is equivalent to $K$-means applied to the expectation of features $\bar{\chi}_i=\mathbb{E}(X_i)$ ($i=1,2,\ldots,n$). 
\citet{Cormode2008approximation} also provides approximation algorithms for the variants of UK-means (e.g., UK-median). 
%(Prob-dist) 
\citet{Kriegel2008density} and \citet{Jiang2013clsutering} define distances between the densities $p_i,p_j$ and apply the simlarity-based clustering methods to the proposed distances (also see Section~\ref{subsec:discussion} for the related approach).

We last note that, clustering uncertain data is distinct from fuzzy clustering (see, e.g., \citet{bezdek1981pattern}), which outputs multiple assignments of clusters with deterministic input.

\begin{figure}[!ht]
    \centering
        \begin{minipage}{0.19\textwidth}
    \fbox{\includegraphics[width=0.95\textwidth]{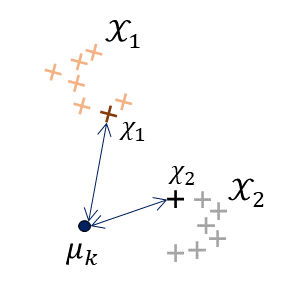}}
    \subcaption{GOC (optimistic)}
    \label{fig:comp_optimistic}
    \end{minipage}
    \hspace{0.1em}
    \begin{minipage}{0.19\textwidth}
    \includegraphics[width=\textwidth]{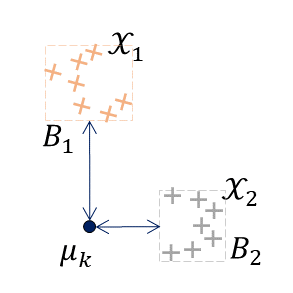}
    \subcaption{Hausdorff+Box}
    \label{fig:comp_Ngai}
    \end{minipage}
    \begin{minipage}{0.19\textwidth}
    \includegraphics[width=\textwidth]{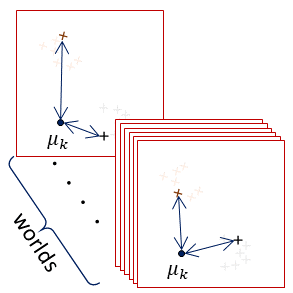}
    \subcaption{PW model}
    \label{fig:comp_PW}
    \end{minipage}
    \begin{minipage}{0.19\textwidth}
    \includegraphics[width=\textwidth]{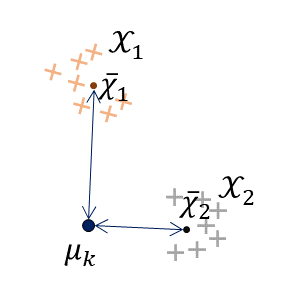}
    \subcaption{UK-means}
    \label{fig:comp_UK}
    \end{minipage}
    \begin{minipage}{0.19\textwidth}
    \includegraphics[width=\textwidth]{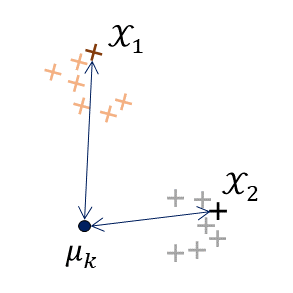}
    \subcaption{Pessimistic}
    \label{fig:comp_pessimistic}
    \end{minipage}
    \caption{Comparison of how to measure the distance from the cluster center $\mu_k$ (to the uncertainty set $\setX_i$). 
    (\subref{fig:comp_optimistic}) Proposed GOC considers the distance to the nearest instance, 
    (\subref{fig:comp_Ngai}) \citet{Ngai2006efficient} computes the distance to the boundary of the minimum box $B_i (\supset \setX_i)$, 
    (\subref{fig:comp_PW}) PW model considers all the possible worlds $\Xi_1,\Xi_2,\cdots \in \setA_n$, 
    (\subref{fig:comp_UK}) UK-means is equivalent to applying $K$-means to $\bar{\chi}_i=\mathbb{E}(X_i)$, 
    (\subref{fig:comp_pessimistic}) robust optimization applied to clustering (a slight modification of \citet{vo2016robust}) considers the distance to the farthest instance. }
\end{figure}

\section{Problem Setting}
\label{sec:problem_setting}

Let $n,d \in \mathbb{N}$; in addition, let $[n]$ denote a set $\{1,2,\ldots,n\}$, where 
$[n]$ also denotes the set of individuals to be clustered. 
Assume that individual $i \in [n]$ is associated with covariate $Z_i \in \mathbb{R}^{d}$ following a known distribution $\mathbb{P}_{Z_i}$. 
Typically, we assume a normal distribution $\mathbb{P}_{Z_i} = \normaldist(\zeta_i,\Sigma)$ with the observed covariate instance $\zeta_i \in \mathbb{R}^d$ of individual $i \in [n]$ and the positive-definite variance-covariance matrix $\Sigma \in \mathbb{R}^{d \times d}$. 
Thus, it is reasonable to employ a covariate uncertainty set $\mathcal{Z}_i \subset \mathbb{R}^{d}$ such that
\begin{align}
    \mathbb{P}(Z_i \in \setZ_i) \ge 1-\eta
    \label{eq:z_probability}
\end{align}
with a user-specified small threshold parameter $\eta > 0$. For the uncertainty set, we employ an upper-level set of the probability density function $p_i$, i.e.,
\begin{align}
\setZ_i:=\{Z \in \mathbb{R}^{d} \mid p_i(Z)>\varepsilon\},
\label{eq:uncertainty_set_upper_level}
\end{align}
for some $\varepsilon=\varepsilon(\eta)>0$, satisfying the inequality~(\ref{eq:z_probability}). 
For instance, by assuming that 
$p_i$ is a Laplace distribution $p_i(Z) \propto \exp(-\lambda \|Z - \zeta_i\|_1)$, the set~(\ref{eq:uncertainty_set_upper_level}) reduces to the box-type uncertainty set $\{Z \mid \|Z-\zeta_i\|_1 \le \rho \}$ used in \citet{vo2016robust}, and by assuming that $p_i$ is a standard normal distribution, $p_i(Z) \propto \exp(-\lambda \|Z-\zeta_i\|^2)$
(\ref{eq:uncertainty_set_upper_level}) reduces to the Euclidean ball $\{Z \mid \|Z-\zeta_i\|_2 \le \rho \}$ used in \citet{ben2002robust}. 
In general, an ellipsoid $\{Z \mid \langle Z-\zeta_i,\Sigma_i^{-1}(Z-\zeta_i)\rangle \le \rho\}$ equipped with a positive definite matrix $\Sigma_i \in \mathbb{R}^{d \times d}$ can be obtained by assuming a normal distribution $p_i(Z) \propto \exp(-\lambda (Z-\zeta_i, \Sigma_i^{-1} (Z-\zeta_i) \rangle)$. 
Herein, we assume that set $\setZ_i$ is specified in advance for each individual $i \in [n]$.

To capture the latent data structure, we further consider a pre-processing, that is, the application of a nonlinear transformation $f:\mathbb{R}^{d} \to \mathbb{R}^{q}$ into $Z_i$. 
This pre-processing step is also expected to remove redundant information that is harmful for a clustering analysis. 
An uncertainty set for this pre-processed feature $X_i=f(Z_i) \in \mathbb{R}^q$ can be expressed as 
\begin{align}
    \setX_i=f(\setZ_i):=\{f(Z) \mid Z \in \setZ_i\} \subset \mathbb{R}^{d},
    \label{eq:feature_ambiguity_set}
\end{align}
which satisfies the following probability inequality:
\[
\mathbb{P}(X_i \in \setX_i)
\ge 
\mathbb{P}(Z_i \in \setZ_i)
\overset{(\ref{eq:z_probability})}{\ge}
1-\eta.
\]

Interestingly, even if the covariate uncertainty set $\setZ_i$ is convex, the feature uncertainty set (\ref{eq:feature_ambiguity_set}) does not necessarily inherit the convexity.

Given (1) covariate uncertainty sets $\setZ_1,\setZ_2,\ldots,\setZ_n$ for individuals to be clustered, and (2) the nonlinear transformation $f:\mathbb{R}^d \to \mathbb{R}^q$ for pre-processing, our goal is to cluster the individuals $[n]$ by exploiting the feature uncertainty sets $\setX_1=f(\setZ_1),\setX_2=f(\setZ_2),\ldots,\setX_n=f(\setZ_n)$. 
It is possible for the feature uncertainty sets to be non-identical and non-convex, and thus their theoretically explicit forms are difficult to obtain. 
In this paper, we provide two examples.

\begin{ex} 
\label{ex:stars}
In galactic astronomy, it is important to identify groups of stars (\emph{sibling stars}) with similar orbits in the Milky Way \citep{roederer2018kinematics,Helmi2020ARA&A..58..205H}. 
Sibling stars were born in the same dwarf galaxies that were later disrupted and absorbed by the Milky Way; 
and therefore sibling stars provide an important insight into the history of the Milky Way. 
Because an orbital period is typically $\sim 10^{8}-10^9$ years, 
humans cannot monitor the entire orbit of a star.
Instead, to find sibling stars, we estimate the orbital properties from 
the instantaneous position $\zeta_i^{(1)} \in \mathbb{R}^3$ (which is an instance of the covariate $Z^{(1)}_i \in \mathbb{R}^3$) and velocity $\zeta_i^{(2)} \in \mathbb{R}^3$ (which is an instance of $Z^{(2)}_i \in \mathbb{R}^3$) of each star 
at the current epoch, by assuming the Galactic gravitational potential. 
See Figure~\ref{fig:KH1}(\subref{fig:zeta123456}) for the six elements in $\zeta_i=(\zeta_i^{(1)},\zeta_i^{(2)}) \in \mathbb{R}^6$. 
Usually, we apply a nonlinear function $f$ to $\zeta_i \in \mathbb{R}^6$ to derive a three-dimensional quantity called an \emph{orbital action} $\chi_i=f(\zeta_i) \in \mathbb{R}^3$, that encapsulates the stellar orbital properties \citep{Binney2008}. 
Whereas the position $\zeta^{(1)}_i$ and velocity $\zeta^{(2)}_i$ change as a function of time, %(i.e., the development of the galaxy), 
the orbital action feature $\chi_i$ is a conserved quantity (see Figure~\ref{fig:KH_visualize_simulation} in Appendix~\ref{subsec:visualization_of_the_simulation} for an illustration of the conserved orbital actions). 
Because sibling stars have similar orbital actions, we apply clustering to the pre-processed feature $\{\chi_i\}_{i=1}^{n}$ instead of the direct observation $\{\zeta_i\}_{i=1}^{n}$. 
Note that the function $f$ omits the remaining information on the instantaneous \emph{orbital phase}, which does not help in finding sibling stars. 

\end{ex}

\begin{figure}[!ht]
\centering
\begin{minipage}{0.49\textwidth}
\centering
\includegraphics[width=0.75\textwidth]{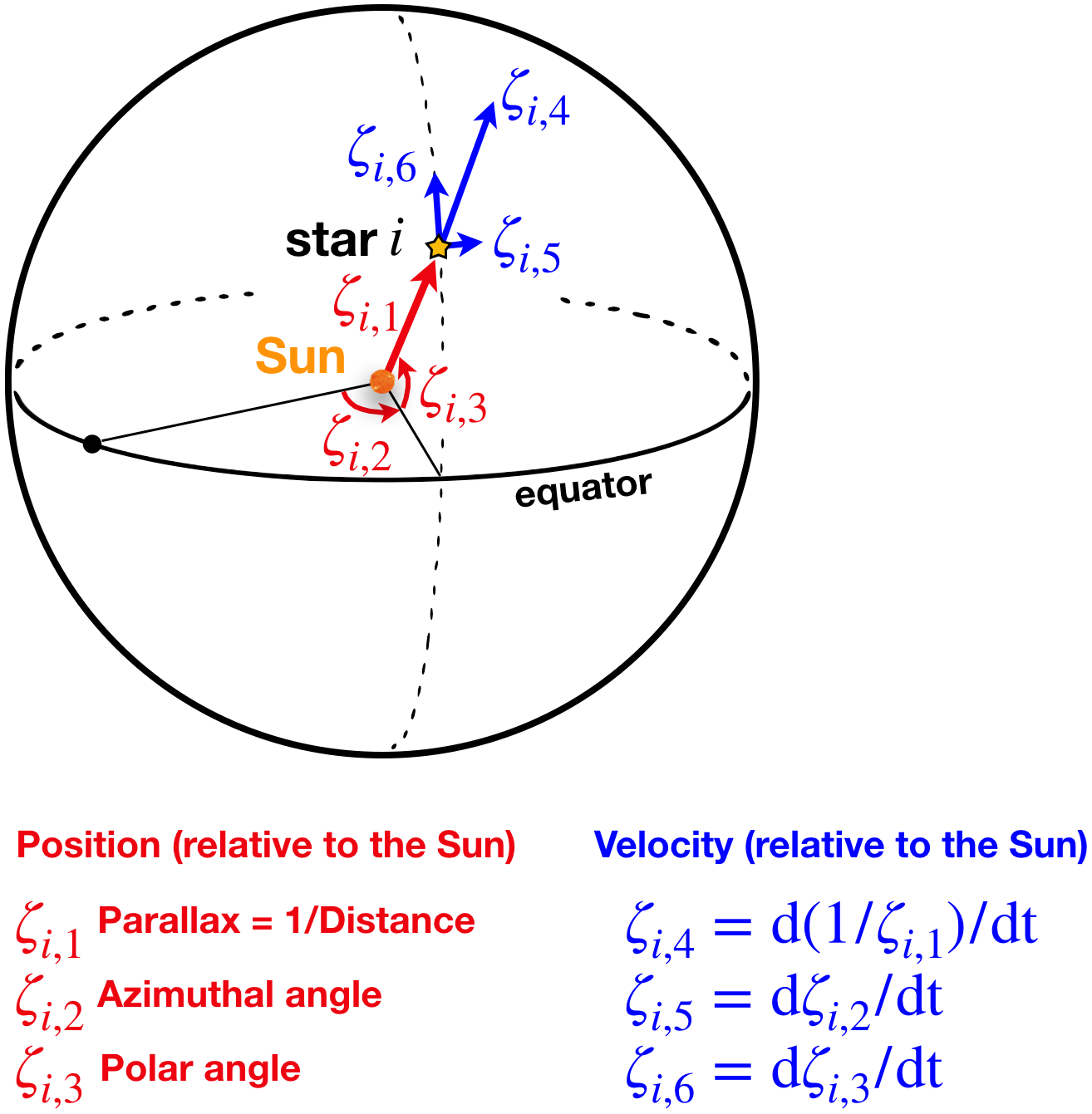}
\subcaption{Stellar position and velocity observed from the Sun.}
\label{fig:zeta123456}
\end{minipage}
%%=============================================================
\begin{minipage}{0.49\textwidth}
\centering
\includegraphics[width=0.75\textwidth]{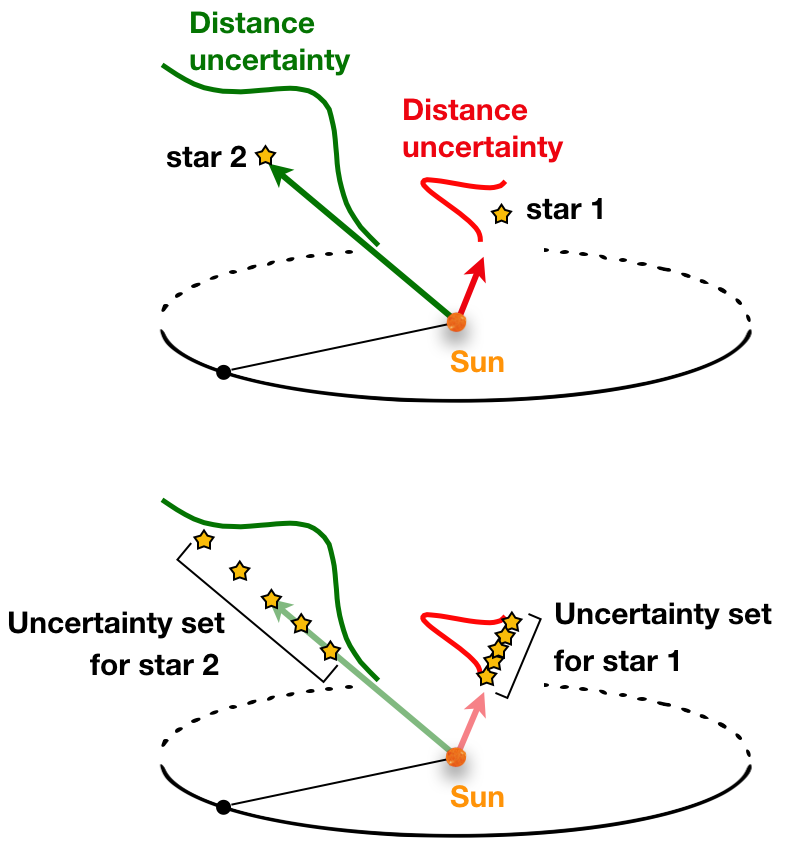}
\subcaption{Uncertainty sets for individual stars.}
\label{fig:uncertainty}
\end{minipage}
%%=============================================================
\caption{An illustration of Example~\ref{ex:stars}.}
\label{fig:KH1}
\end{figure}

Interestingly, in Example~\ref{ex:stars}, the covariate $Z_i^{(1)}$ for the instantaneous position $\zeta_i^{(1)}$ of star $i$ may have a larger uncertainty depending on the stellar properties (e.g., fainter stars have a larger uncertainty) or the cadence of the observations (e.g., stars observed less frequently have a larger uncertainty). 
See Figure~\ref{fig:KH1}(\subref{fig:uncertainty}). 
Therefore, the scatter of the uncertainty set depends on the individual star $i \in [n]$. 
We provide synthetic datasets demonstrating Example~\ref{ex:stars}. For further details, see Section~\ref{subsec:synthetic_dataset_generation} and Appendix~\ref{subsec:detailed_descriptions_of_synthetic_dataset}.

\begin{ex} 
\label{ex:dimensionality_reduction}
For the general clustering problem of individuals $[n]$ equipped with observed covariate instances $\zeta_i \in \mathbb{R}^d$ ($i \in [n]$), we can assume that $\zeta_i$ is an instance of the covariate $Z_i \in \mathbb{R}^{d}$ following the distribution $\mathbb{P}_{Z_i}$, which is typically a normal distribution $\mathbb{P}_{Z_i}=\normaldist(\zeta_i,\Sigma_i)$ equipped with a positive-definite variance-covariance matrix $\Sigma_i \in \mathbb{R}^{d \times d}$. 
The uncertainty set, $\setZ_i$, can be specified by (\ref{eq:uncertainty_set_upper_level}). 
To cluster individuals $i \in [n]$, a nonlinear dimensionality reduction, including a kernel principal component analysis~\citep{scholkopf1998nonlinear}, can be applied beforehand to the observed covariate instances $\{\zeta_i\}_{i=1}^{n}$. 
Using the nonlinear transformation $f:\mathbb{R}^{d}\to\mathbb{R}^{q}$ for a dimensionality reduction ($q<d$), we obtain the pre-processed feature instance $\chi_i = f(\zeta_i)$, and the feature uncertainty set can be specified by $\setX_i=f(\setZ_i)$.  
\end{ex}

\section{Greedy and Optimistic Clustering Algorithm}

In Section~\ref{subsec:empirical_uncertainty_set}, we first define an empirical set that approximates feature uncertainty set $\setX_i$ through a synthetic data generation over the specified covariate uncertainty set $\setZ_i$. 
Using this empirical set, in Section~\ref{subsec:greedy_optimistic_clustering}, we propose the GOC algorithm for clustering with optimism,
which exploits a user-specified (arbitrary) clustering oracle $\clustering$.  The proposed algorithm equipped with a simple $K$-means clustering is further discussed in Section~\ref{subsec:GOC_kmeans}.

\subsection{Empirical Feature Uncertainty Set}
\label{subsec:empirical_uncertainty_set}
Because of the nonlinear function $f:\mathbb{R}^d \to \mathbb{R}^q$ used in the pre-processing, the theoretically explicit form of the uncertainty set $\setX_i=f(\setZ_i)$ for the pre-processed feature $X_i=f(Z_i)$ is difficult to obtain. 
Therefore, we approximate set $\setX_i$ through synthetic data generation using the user-specified set $\setZ_i$ and nonlinear function $f$.

To approximate set $\setX_i$, we employ $m_i$ instances of the $d$-dimensional random variable following a uniform distribution over set $\setZ_i \subset \mathbb{R}^{d}$, i.e.,
\begin{align}
    \zeta_{i}^{(1)},\zeta_{i}^{(2)},\ldots,\zeta_{i}^{(m_i)} 
    \text{ are instances i.i.d. drawn from }
    \text{Unif.}(\setZ_i),
\label{eq:draw_Z}
\end{align}
and define an empirical feature uncertainty set as follows:
\begin{align}
    \tilde{\setX}_i^{(m_i)}:=\{\chi_{i}^{(j)}\}_{j=1}^{m_i} \subset \mathbb{R}^{q}, 
    \: 
    \chi_{i}^{(j)} = f(\zeta_{i}^{(j)}), 
    \quad 
    (j=1,2,\ldots,m_i).
    \label{eq:empirical_uncertainty_set}
\end{align}
$\{m_i\}_{i=1}^{n} \subset \mathbb{N}$ are hyperparameters, typically, $m_i = 100$. 
By assuming the following non-degenerate condition on function $f$ with the Lebesgue measure $\mathcal{L}$ (over the $d$-dimensional Euclidean space $\mathbb{R}^{d}$), i.e.,
\begin{align}
    \inf_{\chi \in \setX_i} \mathcal{L}(\{Z \in \setZ_i \mid \|f(Z)-\chi\|_2 < \varepsilon\})>0, 
    \quad \forall \varepsilon>0,
\label{eq:non-degenerate}
\end{align}

proposition~\ref{prop:uncertainty_set_approximation} proves that the set $\tilde{\setX}_i^{(m_i)}$ approximates the uncertainty set $\setX_i$.

\begin{prop}
\label{prop:uncertainty_set_approximation}
Assume that $f:\mathbb{R}^{d} \to \mathbb{R}^{q}$ satisfies the non-degenerate condition~(\ref{eq:non-degenerate}). 
For any $\chi \in \setX_i$, it holds that 
$\min_{\chi' \in \tilde{\setX}_i^{(m_i)}} \|\chi-\chi'\|_2 \to 0$ in probability, for $m_i \to \infty$. 
\end{prop} 

\begin{proof} 
Let $B_{\varepsilon}(\chi)$ be a ball centered at $\chi \in \mathbb{R}^q$ with radius $\varepsilon>0$; in addition, let 
$\overline{B}=\mathbb{R}^q \setminus B$ denote the complement of set $B \subset \mathbb{R}^q$. 
This assertion is proved by $\mathbb{P}(\min_{\chi' \in \tilde{\setX}_i^{(m_i)}}\|\chi-\chi'\|_2 > \varepsilon)=\mathbb{P}(\tilde{\setX}_i^{(m_i)} \subset \overline{B_{\varepsilon}(\chi)})=\{1-\mathbb{P}(X_i \in B_{\varepsilon}(\chi)) \}^{m_i}=:(1-c)^{m_i} \to 0$ for $m_i \to \infty$ for any $\varepsilon>0$ because 
$c>0$ follows from the condition~(\ref{eq:non-degenerate}).
\end{proof}

Taking an arbitrary feature instance $\chi \in \setX_i$, the set $\tilde{\setX}_i^{(m_i)}$ with a sufficiently large $m_i \in \mathbb{N}$ has an entry $\chi' \in \tilde{\setX}_i^{(m_i)}$ sufficiently close to $\chi$, 
for which we can expect that 
\begin{align}
\argmin_{\chi \in \tilde{\setX}_i^{(m_i)}}L(\chi) 
\, 
\approx
\, 
\argmin_{\chi \in \setX_i}L(\chi)
\label{eq:approximation_of_argmin}
\end{align}
for a wide class of functions $L:\mathbb{R}^q \to \mathbb{R}_{\ge 0}$. 
This proposition can be applied to (\ref{eq:update_X}) in GOC (under some mild assumptions).

\subsection{Greedy and Optimistic Clustering Algorithm }
\label{subsec:greedy_optimistic_clustering}

In this section, we propose a greedy and optimistic clustering~(GOC) algorithm, that greedily seeks the feature candidates $\Xi=(\chi_1,\chi_2,\ldots,\chi_n)$ over a set 
\[
   \setA_n := \tilde{\setX}_1^{(m_1)} \times \tilde{\setX}_2^{(m_2)} \times \cdots \tilde{\setX}_n^{(m_n)},
\]
yielding more condensed clusters. 
Given the initial feature candidates $\Xi(0)=(\xi_1(0),\xi_2(0),\ldots,\xi_n(0)) \in \setA_n$, initial number of clusters $K(0) \in [n]$, initial cluster assignments $\hat{\bs c}(0) \in [K(0)]^n$, 
and an user-specified \emph{arbitrary} clustering oracle $\clustering(\Xi,K)$ (e.g., $K$-means), which outputs the cluster assignments of individuals $[n]$ by taking an instance $\Xi \in \setA_n$ (as well, we can input the number of clusters $K \in \mathbb{N}$, initial cluster centers, and some additional parameters, if necessary), GOC iterates the following steps: at iteration $t=1,2,\ldots,T$, 

\begin{enumerate}[{(I)}]

\item The designated number of clusters is updated. 
Typically, we may employ a constant $K(t) = K$ or the number of clusters found in the previous step $K(t) = \sum_{k=1}^{K(t-1)}\mathbbm{1}\left( \sum_{i=1}^{n}\mathbbm{1}(\hat{c}_i(t-1)=k)>1 \right)$.

\item Temporal cluster assignments $\hat{\bs c}^{\dagger}(t) \in [K(t)]^n$ are obtained using the clustering oracle:
\[
\hat{\bs c}^{\dagger}(t) \leftarrow \clustering(\Xi (t-1);K(t)). 
\] 

\item Feature candidates and cluster assignments are updated to $\Xi(t)=(\chi_1(t),\chi_2(t),\ldots,\chi_n(t))$ and $\hat{\bs c}(t)=(\hat{c}_1(t),\hat{c}_2(t),\ldots,\hat{c}_n(t))$, respectively, by applying
\begin{align} 
\chi_i(t)
&\leftarrow 
\argmin_{\chi \in \tilde{\setX}_i^{(m_i)}} 
\min_{k \in [K(t)]}
\left\{ \|\chi - \hat{\mu}_k(\Xi(t-1),\hat{\bs c}^{\dagger}(t))\|_2^2 + \lambda \text{Pen}_i(\chi) \right\}, \label{eq:update_X} \\
\hat{c}_i(t)
&\leftarrow 
\argmin_{k \in [K(t)]}
\|\chi_i(t)-\hat{\mu}_k(\Xi(t-1),\hat{\bs c}^{\dagger}(t))\|_2
\nonumber 
\end{align}
$(i \in [n])$, where 
\begin{align}
\hat{\mu}_k(\Xi,\hat{\bs c}) := \frac{\sum_{i=1}^{n} \mathbb{1}(\hat{c}_i=k)\chi_i}{\sum_{i=1}^{n} \mathbb{1}(\hat{c}_i=k)} \in \mathbb{R}^d
\label{eq:mu_k}
\end{align}
denotes the cluster center, $\text{Pen}_i(X)$ denotes a user-specified penalty term (e.g., $\text{Pen}_i(X):=\|X-\chi_i'\|_2^2$ for some $\chi_i' \in \mathbb{R}^q$), and $\lambda \ge 0$ is a hyperparameter. 
Note that the penalty term $\text{Pen}_i$ may depend on the individual $i=1,2,\ldots,n$. 
\end{enumerate}

Steps (I)--(III) are repeated until convergence is reached. 
See Figure~\ref{fig:illustration_of_GOC} for an illustration, and Section~\ref{subsec:GOC_kmeans} for an interpretation of the GOC algorithm using $K$-means clustering.

\begin{figure}[!ht]
\centering
\begin{minipage}{0.32\textwidth}
\centering
\includegraphics[width=\textwidth]{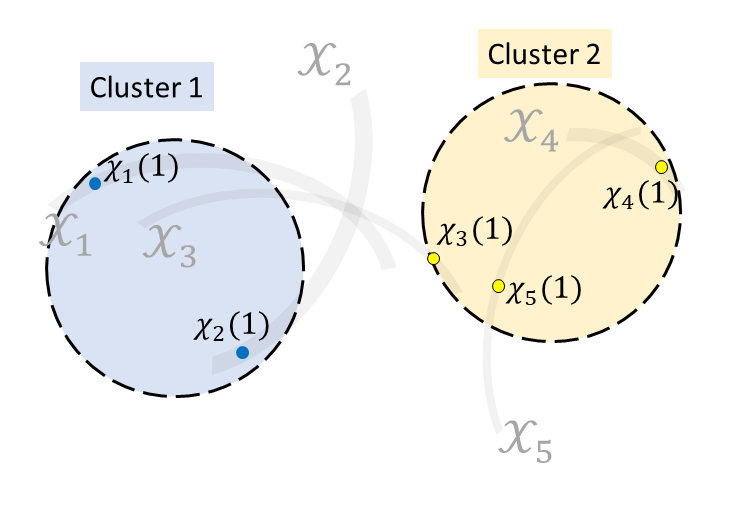}
\subcaption{Step (II): Temporal cluster assignments are obtained as $\hat{c}_1^{\dagger}(2)=\hat{c}_2^{\dagger}(2)=1,\hat{c}_3^{\dagger}(2)=\hat{c}_4^{\dagger}(2)=\hat{c}_5^{\dagger}(2)=2$.}
\end{minipage}
\begin{minipage}{0.32\textwidth}
\centering
\includegraphics[width=\textwidth]{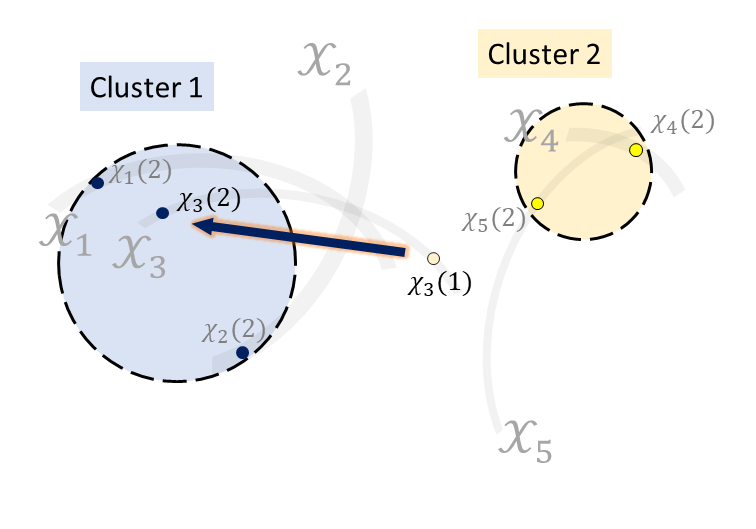}
\subcaption{Step (III): Feature candidates and cluster assignments are updated  ($\hat{c}_3^{\dagger}(2)=2$ is reassigned to $\hat{c}_3(2)=1$).}
\end{minipage}
\hspace{1em}
\begin{minipage}{0.32\textwidth}
\centering
\includegraphics[width=\textwidth]{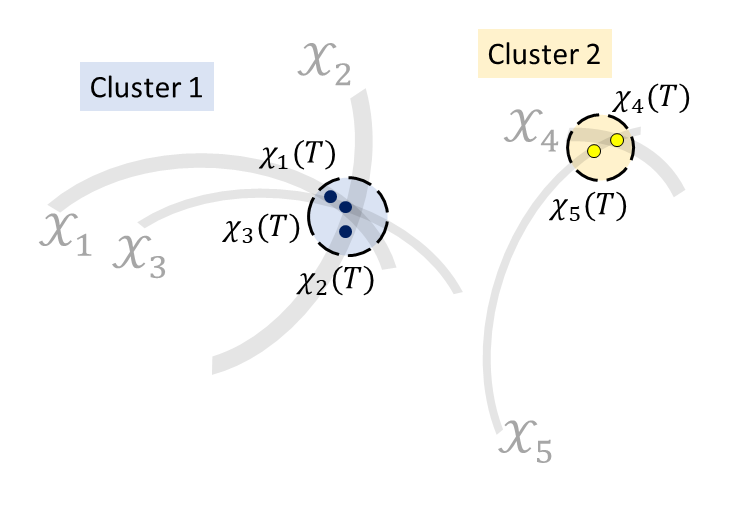}
\subcaption{After a sufficiently large number of iterations $T \in \mathbb{N}$, we expect to obtain the condensed clusters.}
\end{minipage}
\caption{An illustration of the GOC algorithm (at the iteration $t=2$), which iteratively updates the cluster assignments and feature candidates selected from the feature uncertainty sets colored in grey.}
\label{fig:illustration_of_GOC}
\end{figure}

\paragraph{Computational complexity.} 
Step (II) requires the computational complexity to be exactly the same as the clustering oracle $\clustering$ (typically, $O(n^2)$), and step (III) can be solved through a linear search, which requires complexity $O(m_i K(t-1))=O(m_i K(0))$ (typically smaller than $O(n^2)$). Therefore, the overall complexity is approximately $O(T(n^2+m_iK(t-1)))=O(Tn^2)$ with $T \in \mathbb{N}$ iterations, where $T = 10$ is a sufficient number in our numerical experiments, as demonstrated in Experiment~3 described in Section~\ref{sec:experiments}.

\subsection{Interpretation of GOC Algorithm with $K$-means Clustering}
\label{subsec:GOC_kmeans}

In this section, we interpret the results of the proposed GOC algorithm using $K$-means clustering~\citep{macqueen1967some} with a constant number of clusters $K(t) = K$. 
Let 
\[ 
\mathcal{S}_{K,n}:=\left\{ 
\{s_{ik}\}_{i \in [n],k \in [K]} \subset \{0,1\} \, \mid \, \sum_{k=1}^{K} s_{ik}=1, \: \forall i \in [n]
\right\}
\]
be a set of cluster-assignment indicators $\bs s=\{s_{ik}\}$, where $s_{ik}=1$ denotes that individual $i$ is assigned to cluster $k$ (and $0$ otherwise). 
Using a loss function 
\[
    \ell(\Xi;\bs s)
    :=
    \min_{\mu_1,\mu_2,\ldots,\mu_K \in \mathbb{R}^q}
    \sum_{k=1}^{K}
    \sum_{i=1}^{n}
    s_{ik}
    \| \chi_i - \mu_k \|_2^2,
\]
conventional $K$-means clustering applied to a fixed instance $\Xi=(\chi_1,\chi_2,\ldots,\chi_n) \in \setA_n$ computes the cluster assignments $\hat{\bs c}=(\hat{c}_1,\ldots,\hat{c}_n) \in [K]^n$ by solving the following problem: 
\[
    (\text{Conventional}) \qquad 
    \hat{c}_i:=\argmax_{k \in [K]}s_{ik},
    \quad 
    \hat{\bs s}
    :=
    \argmin_{\bs s \in \mathcal{S}_{K,n}}
    \left\{
    \ell(\Xi;\bs s)
    \right\}.
\]

Further, the GOC algorithm equipped with $K$-means clustering is expected to solve the following minimization problem in a greedily manner:
\[
(\text{Optimistic}) \qquad 
\hat{c}_i^{(\text{GOC})} = \argmax_{k \in [K]} \hat{s}_{ik}^{(\text{GOC})},
\quad
\hat{\bs s}^{(\text{GOC})}
:=
\argmin_{\bs s \in \mathcal{S}_{K,n}}
\min_{\tilde{\Xi} \in \setA_n}
\left\{ 
   \ell(\tilde{\Xi};\bs s)
    +
    \lambda \sum_{i=1}^{n}\text{Pen}_i(\tilde{\chi}_i)
\right\}
\]
for $\tilde{\Xi}=(\tilde{\chi}_1,\tilde{\chi}_2,\ldots,\tilde{\chi}_n)$. 
By contrast, we may consider a pessimistic variant of the GOC algorithm, i.e., the greedy and pessimistic clustering (GPC), when solving the following problem:
\[
(\text{Pessimistic}) \qquad 
\hat{c}_i^{(\text{GPC})} = \argmax_{k \in [K]} \hat{s}_{ik}^{(\text{GPC})},
\quad
    \hat{\bs s}^{(\text{GPC})}
    :=
\argmin_{\bs s \in \mathcal{S}_{K,n}}
\max_{\tilde{\Xi} \in \setA_n}
\left\{ 
   \ell(\tilde{\Xi};\bs s)
    +
    \lambda \sum_{i=1}^{n}\text{Pen}_i(\tilde{\chi}_i)
\right\},
\]
the formulation of which achieves a robust optimization~\citep{ben2002robust,bertsimas2011theory}, minimizing the worst-case of the loss function. 
Furthermore, assuming that $\lambda = 0$ (i.e., no penalty for the candidates $\tilde{\Xi} \in \setA_n$) and $\setX_i=\{X \in \mathbb{R}^d \mid \|X-\mu_i\|_1 \le \rho\}$ is a box-shaped convex uncertainty set with a user-specified small threshold parameter $\rho>0$, GPC greedily solves the equivalent problem described in \citet{vo2016robust}, which is the only existing study applying robust optimization to clustering.

However, note that the set $\setX_i$ considered in this study is not the small box-shaped convex set considered in \citet{vo2016robust} (where \citet{vo2016robust} aims to attain robustness against a covariate perturbation but not the larger uncertainty considered herein), and the
GPC applied to our synthetic dataset achieved extremely low scores. 
In particular, because the feature candidates are updated to increase the scattering of each cluster from a pessimistic perspective, regardless of the initial number of clusters $K(0)$, GPC applied to our dataset finally outputs only one large cluster (i.e., $K(T)=1$ after a sufficiently large number of iterations $T$). 
We therefore did not apply GPC in our numerical experiments, as discussed in Section~\ref{sec:experiments}.

\section{Numerical Experiments}
\label{sec:experiments}

In Section~\ref{subsec:synthetic_dataset_generation}, we describe the synthetic orbital action datasets of the stars used in our numerical experiments. 
The experimental settings and results are presented in Sections \ref{subsec:experimental_settings} and \ref{subsec:experimental_results}, respectively. 
Further disccusions are also provided: we consider another approach for exploiting the uncertainty sets in Section~\ref{subsec:discussion} and 
convex feature ambiguity sets in Section~\ref{subsec:discussion2_convex}.
We provide the datasets and \verb|R| source codes used to produce the experimental results at \url{https://github.com/oknakfm/GOC}.

\subsection{Realistic Datasets: Orbital Properties of the Stars in the Milky Way}
\label{subsec:synthetic_dataset_generation}

We employed $10$ synthetic orbital action datasets of stars generated through a numerical simulation mimicking the formation process of the Milky Way. 
Each of these $10$ datasets consists of pre-computed empirical uncertainty sets $\{\tilde{\setX}_1^{(m_1)},\tilde{\setX}_2^{(m_2)},\ldots,\tilde{\setX}_n^{(m_n)}\}$ with $n = 275$ stars, where each star $i \in [n]$ is assigned to one of $K_* = 50$ true clusters and is associated with the empirical uncertainty set $\tilde{\setX}_i^{(m_i)} \subset \mathbb{R}^3$ of size $m_i = 101$. 
The member stars of each true cluster are born in the same dwarf galaxy, which we refer to as {\it sibling stars}. 
Sibling stars have similar orbital actions (i.e., similar orbital properties), and our task is to find sibling stars by leveraging uncertainty sets. 
Figure~\ref{fig:three_instances_of_dataset} shows the first three datasets used.

\begin{figure}[!ht]
\centering
\begin{minipage}{0.33\textwidth}
\centering
\includegraphics[width=0.9\textwidth]{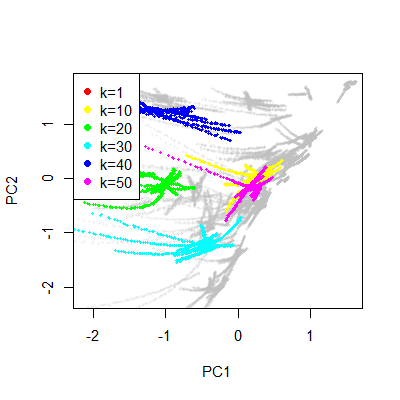}
\subcaption{Instance $1$}
\end{minipage}
%%=============================================================
\begin{minipage}{0.33\textwidth}
\centering
\includegraphics[width=0.9\textwidth]{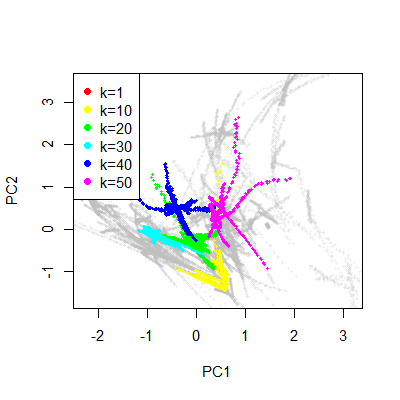}
\subcaption{Instance $2$}
\end{minipage}
%%=============================================================
\begin{minipage}{0.33\textwidth}
\centering
\includegraphics[width=0.9\textwidth]{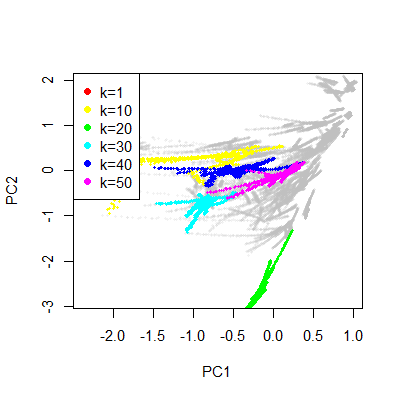}
\subcaption{Instance $3$}
\end{minipage}
\caption{Principal components of the first three datasets. The empirical uncertainty sets of stars belonging to each of the underlying true clusters $k = 1,10,20,30,40,50$ are colored separately. The uncertainty sets for the remaining clusters are colored in gray.}
\label{fig:three_instances_of_dataset}
\end{figure}

In addition to a summary of this dataset, Appendix~\ref{subsec:detailed_descriptions_of_synthetic_dataset} provides more detailed descriptions of the physical simulation.

\bigskip \noindent 
\textbf{Summary of a dataset} is described herein: we generate $10$ different datasets by following the same procedure. 
We generate the uncertainty sets of $n = 275$ stars, which form $K_* = 50$ clusters with similar orbits. 
Following the numerical simulation shown in Appendix~\ref{subsec:detailed_descriptions_of_synthetic_dataset}, we obtain the instance $\zeta_i^*:=(\zeta_i^{*(1)}, \zeta_i^{*(2)})=(\zeta_{i,1}^*,\zeta_{i,2}^*,\ldots,\zeta_{i,6}^*)$,  consisting of the {\it true} current position $\zeta_i^{*(1)}$ and velocity $\zeta_i^{*(2)}$ of the $i$th star ($i \in [n]$) within the observable space. Note that the observable space differs from the usual Cartesian coordinate system. For example, in astronomy, the three-dimensional stellar position $\zeta_i^{*(1)}$ is expressed by the so-called stellar parallax (which is the reciprocal of the stellar distance from Earth ) and the two-dimensional position in the sky. See Figure~\ref{fig:KH1}(\subref{fig:zeta123456}) for an illustration.   
The observational uncertainty in $\zeta^*_{i,\ell}$ is denoted as $\sigma_{i,\ell}$ and is computed from the empirical relationship known in the field of galactic astronomy~(see \citet{GaiaEDR3_2021A&A...649A...1G}). 
Note that, owing to observational difficulty, one of the three components in $\zeta_i^{*(1)}$ (stellar parallax) is associated with a large uncertainty, whereas the other two components in $\zeta_i^{*(1)}$ (two-dimensional position in the sky) and all three components in $\zeta_i^{*(2)}$ are associated with negligible uncertainty.

To mimic an actual observation in which we do not know the {\it true} value of each observable, 
we incorporate randomness into the quantity $\zeta_i^*=(\zeta_{i,1}^*,\ldots,\zeta_{i,6}^*)$, i.e.,
each entry $\zeta_{i,\ell}$ in $\zeta_i=(\zeta_{i,1},\ldots,\zeta_{i,6})$ is drawn independently from $N(\zeta_{i,\ell}^*,\sigma_{i,\ell}^2)$. 
Subsequently, we define the uncertainty set as follows:
\[
    \setZ_i := \{Z=(z_1,z_2,\ldots,z_6) \in \mathbb{R}^6 \mid |z_{\ell}-\zeta_{i,\ell}| \le 2\sigma_{i,\ell}, \ell=1,2,\ldots,6\}
\]
and draw i.i.d. $m_i$ instances $\zeta_i^{(1)},\zeta_i^{(2)},\ldots,\zeta_i^{(m_i)}$ from a uniform distribution $\text{Unif.}(\setZ_i)$ for $i \in [n]$. 
 
Finally, we compute the empirical uncertainty set $\tilde{\setX}_i^{(m_i)}$ using Eq.~(\ref{eq:empirical_uncertainty_set}); that is, 
$\tilde{\setX}_i^{(m_i)}=\{\chi_i^{(j)}\}_{j=1}^{m_i},\chi_i^{(j)}=f(\zeta_i^{(j)}) \in \mathbb{R}^3$. 
In particular, to transform the current positions and velocities of the stars into orbital actions, we employ a publicly available \verb|AGAMA| package~\citep{Vasiliev2018ascl.soft05008V,Vasiliev2019MNRAS.482.1525V} for the nonlinear transformation $f:\mathbb{R}^6 \to \mathbb{R}^3$, 
which effectively removes unnecessary information for finding clusters of sibling stars.

In the simplest implementation of the GOC algorithm, we can treat all elements of the uncertainty set for star $i$ equally. 
However, we can introduce a penalty term to prioritize (penalize) the elements of an uncertainty set that are closer to (farther from) the center of the uncertainty set, $\zeta_i$. 
For this purpose, we define the penalty of the $j$th candidate of star $i$ as 
$\text{Pen}_i(\chi_i^{(j)}) = (\zeta_{i,1}^{(j)} - \zeta_{i,1})^2/2\sigma_{i,1}^2$. 
Note that, because its uncertainty is the dominant source of uncertainty in $\chi_i$, we only consider the first component of $\zeta_i$, which corresponds to the stellar parallax.

\subsection{Experimental Settings} 
\label{subsec:experimental_settings}

\paragraph{Standardization.} 
Before conducting the experiments, we first standardize the empirical uncertainty sets (using both centering and scaling, such that $\sum_{i=1}^{n}\sum_{\zeta \in \tilde{\setX}_i^{(m_i)}}\zeta=0,(\sum_{i=1}^{n}m_i)^{-1}\sum_{i=1}^{n}\sum_{\zeta \in \tilde{\setX}_i^{(m_i)}}\zeta^2=1$) and the penalties (using scaling only, such that $\max_{i \in [n]} \max_{j} \text{Pen}_i(X_i^{j})=1$) for each dataset.

\paragraph{Clustering oracles.} We employ $K$-means~(using the standard \verb|stats| package in \verb|R| statistical software) and 
$K$-medoids~(using the \verb|ClusterR| package), which can take the cluster centers as their input, and for each iteration, we input the cluster center $\hat{\mu}_k=(\ref{eq:mu_k})$. 
We also employ two different implementations of a Gaussian mixture model(GMM), i.e.,
the \verb|GMM| function in the \verb|ClusterR| package, and the \verb|Mclust| function in the \verb|mclust| package, the latter of which (\verb|Mclust|) can specify models for the variance-covariance matrix $\Sigma_k$ of the Gaussian distribution (representing the cluster $k$). 
We employ the simplest EII model~($\Sigma_k=\sigma^2 I$ for a certain $\sigma>0$; more general VII and VVV models are also mentioned in the note of Experiment 4), and \verb|Mclust| automatically detects the number of clusters using BIC  (from $1,2,\ldots,K(t)$). See \citet{R_mclust} for further details.

\paragraph{Baselines.} For the baselines, we apply clustering oracles ($K$-means, $K$-medoids, and GMM) to the vectors
\begin{align}
    \bar{\chi}_i := \frac{1}{m_i}\sum_{\chi \in \tilde{\setX}_i^{(m_i)}} \chi \in \mathbb{R}^3
    \quad (i \in [n]), 
    \label{eq:representative_vector}
\end{align}
which represent the sample mean of each feature uncertainty set. 
Referring to \citet{Lee2006reducing}, the $K$-means applied to (\ref{eq:representative_vector}) also can be regarded as UK-means~\citep{Chau2006uncertain} by assuming that the feature $X_i$ follows a uniform distribution over the set $\tilde{\setX}_i^{(m_i)}$ (for $i \in [n]$).

\paragraph{Evaluation metrics.} 
We define the following scores, using the estimated clusters $\hat{\bs c}=(\hat{c}_1,\hat{c}_2,\ldots,\hat{c}_n) \in [K]^n$ as well as the true clusters $\bs c^*=(c^*_1,c^*_2,\ldots,c^*_n) \in [K_*]^n$, 
$n_{kl}:=\sum_{s=1}^{n}
\mathbbm{1}(\hat{c}_s=k)
\mathbbm{1}(c^*_s=l)$, 
$n_{k \cdot}:=\sum_{l=1}^{K_*}n_{kl}$, and 
$n_{\cdot l}:=\sum_{k=1}^{K}n_{kl}$. 

\begin{enumerate}

    \item Normalized mutual information $\text{NMI}(\hat{\bs c},\bs c^*)$ is defined by $2\mathcal{I}/(\mathcal{H}^{(1)}+\mathcal{H}^{(2)})$, using the mutual information $\mathcal{I}$ and entropy $\mathcal{H}$: \begin{align*} 
    \mathcal{I}
    &:=
    \sum_{k=1}^{K}\sum_{l=1}^{K_*} \frac{n_{kl}}{n}
    \log 
    \frac{n \cdot n_{kl}}{
        n_{\cdot l} \cdot n_{k \cdot}
    },
\quad 
    \mathcal{H}^{(1)}
    :=
    -\sum_{k=1}^{K} \frac{n_{k \cdot}}{n} 
    \log \frac{n_{k \cdot}}{n}, 
\quad 
    \mathcal{H}^{(2)}
    :=
    -\sum_{l=1}^{K_*} \frac{n_{\cdot l}}{n} 
    \log \frac{n_{\cdot l}}{n}.
    \end{align*}

    \item The \textbf{$F$-measure} $F(\hat{\bs c},\bs c_*)$ is defined as 
    \begin{align*} 
        \sum_{l=1}^{K_*}
        \frac{n_{\cdot l}}{n}
        \max_{k \in [K]}
        \left\{
            \left( \frac{n_{kl}}{n_{k \cdot}} + \frac{n_{kl}}{n_{\cdot l}} \right)^{-1}
            \left( \frac{2n_{kl}^2}{n_{k \cdot} n_{\cdot l}} \right)
        \right\}.
    \end{align*}
\end{enumerate}

Both the NMI and $F$-measure take values within $[0,1]$, and attain a value of $1$ if and only if the estimated clusters $\hat{\bs c}$ perfectly match the true clusters $\bs c^*$ (up to the permutation of the cluster labels).

\paragraph{Additional settings.} The clustering step used by the GOC algorithm employs the number of clusters appearing in the previous step, i.e., $K(t)=\sum_{k=1}^{K(t-1)}\mathbbm{1}\left(\sum_{i=1}^{n}\mathbbm{1}(\hat{c}_i(t-1)=k)>1 \right)$, and thus the number of clusters for the algorithm can be smaller than the (user-specified) initial number of clusters $K(0) \in \mathbb{N}$. 
We consider the GOC algorithm to reach convergence if the selected feature candidates are converged. More specifically, in Experiments $1$, $2$, and $4$, we consider the perfect convergence of the feature candidates (selected from the discrete set $\setA_n$); by allowing a tolerance on the convergence, we can terminate GOC with fewer iterations as observed in Experiment~$3$.

\subsection{Experimental Results}
\label{subsec:experimental_results}

We apply the GOC algorithm along with the clustering oracle $\clustering$ to synthetic datasets consisting of empirical uncertainty sets. 
For the baselines, we also apply the oracle $\clustering$ to the representative vectors $\{\bar{\chi}_i\}_{i=1}^{n}$ defined in (\ref{eq:representative_vector}). 
Whereas Experiments 1--3 computed only the GOC algorithm when applying $K$-means (as well as the corresponding baseline, i.e., $K$-means applied to the representative vectors), $K$-medoids and GMM were applied in Experiment 4.

\subsubsection*{Experiment 1: Fixed $K(0)$ with increasing $\lambda$.}

Table~\ref{table:exp_fixed_K_increasing_lambda} shows the NMI and $F$-measure with a fixed initial number of clusters $K(0)$ and increasing $\lambda$ (the coefficient of the penalty term).

\begin{table}[!ht]
\centering
\caption{$K$-means with fixed $K(0)$ and increasing coefficient $\lambda$ for the penalty term.}
\label{table:exp_fixed_K_increasing_lambda}
\subcaption{$K(0)=30$}
\begin{tabular}{llcccc} 
 \toprule 
 & & $\lambda=0$ & $\lambda=0.01$ & $\lambda=0.1$ & $\lambda=1$ \\ 
 \midrule 
 \multirow{2}{*}{NMI} & GOC  & $0.830\pm0.030$ & $0.834\pm0.023$ & $0.844\pm0.030$ & $0.821\pm0.021$ \\ 
 & Baseline  & \multicolumn{4}{c}{$0.808\pm0.020$} \\ 
  \cmidrule{2-6} 
\multirow{2}{*}{$F$-measure} & GOC  & $0.627\pm0.042$ & $0.641\pm0.026$ & $0.648\pm0.051$ & $0.618\pm0.027$ \\ 
 & Baseline  & \multicolumn{4}{c}{$0.604\pm0.024$}, \\ 
 \cmidrule{2-6} 
$\#$clusters & GOC & $28.3\pm1.34$ & $28.4\pm0.84$ & $28.5\pm0.85$ & $28.9\pm0.57$ \\ 
 $\#$iterations & GOC & $15.2\pm3.05$ & $15.8\pm3.74$ & $13.3\pm4.40$ & $7.3\pm1.70$ \\ 
 \bottomrule 
 \end{tabular}
%%=========================================================================
%%=========================================================================
\subcaption{$K(0)=50$}
\begin{tabular}{llcccc} 
 \toprule 
 & & $\lambda=0$ & $\lambda=0.01$ & $\lambda=0.1$ & $\lambda=1$ \\ 
 \midrule 
 \multirow{2}{*}{NMI} & GOC  & $0.880\pm0.024$ & $0.879\pm0.027$ & $0.871\pm0.024$ & $0.846\pm0.026$ \\ 
 & Baseline & \multicolumn{4}{c}{$0.839\pm0.026$}, \\ 
  \cmidrule{2-6} 
\multirow{2}{*}{$F$-measure} & GOC  & $0.750\pm0.039$ & $0.752\pm0.046$ & $0.736\pm0.041$ & $0.694\pm0.045$ \\ 
 & Baseline  & \multicolumn{4}{c}{$0.685\pm0.048$} \\ 
 \cmidrule{2-6} 
$\#$clusters & GOC & $46.1\pm1.10$ & $46.4\pm0.97$ & $47.2\pm1.23$ & $48.5\pm0.85$ \\ 
 $\#$iterations & GOC & $15.8\pm3.23$ & $15\pm3.33$ & $12.8\pm3.74$ & $7.4\pm1.84$ \\ 
 \bottomrule 
 \end{tabular}
%%=========================================================================
%%=========================================================================
\subcaption{$K(0) = 70$}
\begin{tabular}{llcccc} 
 \toprule 
 & & $\lambda=0$ & $\lambda=0.01$ & $\lambda=0.1$ & $\lambda=1$ \\ 
 \midrule 
 \multirow{2}{*}{NMI} & GOC & $0.877\pm0.020$ & $0.878\pm0.022$ & $0.875\pm0.018$ & $0.850\pm0.027$ \\ 
 & Baseline  & \multicolumn{4}{c}{$0.837\pm0.021$} \\ 
  \cmidrule{2-6} 
\multirow{2}{*}{$F$-measure} & GOC & $0.749\pm0.035$ & $0.747\pm0.036$ & $0.742\pm0.032$ & $0.706\pm0.051$ \\ 
 & Baseline  & \multicolumn{4}{c}{$0.673\pm0.035$} \\ 
 \cmidrule{2-6} 
$\#$clusters & GOC & $64.6\pm1.71$ & $64.2\pm1.55$ & $66.6\pm1.51$ & $67.8\pm1.32$ \\ 
 $\#$iterations & GOC & $11.7\pm1.50$ & $12.2\pm2.10$ & $11.9\pm2.03$ & $7.7\pm1.06$ \\ 
 \bottomrule 
 \end{tabular}
\end{table}

\paragraph{Results.} 
First, for the experiments with a small $\lambda$, the proposed GOC algorithm outperforms the baselines. Employing feature uncertainty sets also improves the NMI and $F$-measure. 
A large coefficient $\lambda$ of the penalty term degrades both the NMI and the $F$-measure; in addition, clustering with stars having a smaller penalty does not necessarily improve the clustering scores. 
Although the number of clusters used by the GOC algorithm can be small, the number is not significantly different from the initial number of clusters $K(0)$.

\subsubsection*{Experiment 2: Fixed $\lambda$ with increasing $K$.}
Table~\ref{table:exp_fixed_lambda_increasing_K} shows the NMI and $F$-measure with a fixed $\lambda$ and increased number of initial clusters $K(0)$.

\begin{table}[!ht]
\centering
\caption{$K$-means with fixed $\lambda=0.01$ and increasing $K(0)$.}
\label{table:exp_fixed_lambda_increasing_K} 
\begin{tabular}{llccccc} 
 \toprule 
 & & $K(0)=30$ & $K(0)=40$ & $K(0)=50$ & $K(0)=60$ & $K(0)=70$ \\ 
 \midrule 
 \multirow{2}{*}{NMI} & GOC  & $0.834\pm0.023$ & $0.868\pm0.023$ & $0.879\pm0.027$ & $0.874\pm0.027$ & $0.878\pm0.022$ \\ 
 & Baseline & $0.808\pm0.020$ & $0.832\pm0.026$ & $0.839\pm0.026$ & $0.839\pm0.027$ & $0.837\pm0.021$ \\ 
 \cmidrule{2-7} 
\multirow{2}{*}{$F$-measure} & GOC & $0.641\pm0.026$ & $0.719\pm0.037$ & $0.752\pm0.046$ & $0.741\pm0.044$ & $0.747\pm0.036$ \\ 
 & Baseline  & $0.604\pm0.024$ & $0.671\pm0.039$ & $0.685\pm0.048$ & $0.682\pm0.047$ & $0.673\pm0.035$ \\ 
 \cmidrule{2-7} 
$\#$clusters & GOC & $28.4\pm0.84$ & $37.6\pm0.84$ & $46.4\pm0.97$ & $55.9\pm1.52$ & $64.2\pm1.55$ \\ 
 $\#$iterations & GOC & $15.8\pm3.74$ & $15.5\pm5.32$ & $15.0\pm3.33$ & $13.5\pm1.96$ & $12.2\pm2.10$ \\ 
 \bottomrule 
 \end{tabular}
\end{table}

\paragraph{Results.} 
Because the true number of clusters is $K_* = 50$, both the GOC algorithm and the baselines demonstrate a better performance with $K(0) \ge 50$ than with $K(0) = 30,40$. 
Overall, with $K(0) \in \{30,40,50,60,70\}$, fewer iterations are required (for GOC convergence) for a larger $K(0)$, and a larger $K(0)$ tends to result in higher scores.

\subsubsection*{Experiment 3: Convergence.}
Although the convergence of the GOC algorithm is determined by the \emph{perfect} convergence of the feature candidates in Experiments $1$ and $2$ (because the candidates $\Xi(t)$ are selected over the discrete set $\setA_n$), 
we monitored the convergence of the GOC algorithm in a weaker sense. 
In particular, we monitored the convergences of (1) the cluster assignments and (2) the feature candidates. 
To evaluate the convergence, for each dataset, we computed two scores $\eta_1(t):=\text{NMI}(\hat{\bs c}(t),\hat{\bs c}(\infty))$ and $\eta_2(t):=n^{-1}\sum_{i=1}^{n} \|\chi_i(t)-\chi_i(\infty)\|_2^2$, where $\hat{\bs c}(\infty),\{\chi_i(\infty)\}_{i=1}^{n}$ denote the final cluster assignments and feature candidates of the GOC algorithm, respectively. 
We also computed the convergence of the NMI scores to the underlying true clusters by evaluating $\eta_3(t):=\text{NMI}(\hat{\bs c}(t),\bs c^*)/\text{NMI}(\hat{\bs c}(\infty),\bs c^*)$. 
Figure~\ref{fig:convergence} shows these scores (for each iteration $t$) for $10$ datasets for $K(0) \in \{30,50,70\}$ with a fixed $\lambda=0.01$. 

\begin{figure}[!ht]
\centering
\begin{minipage}{0.33\textwidth}
\centering
\includegraphics[width=\textwidth]{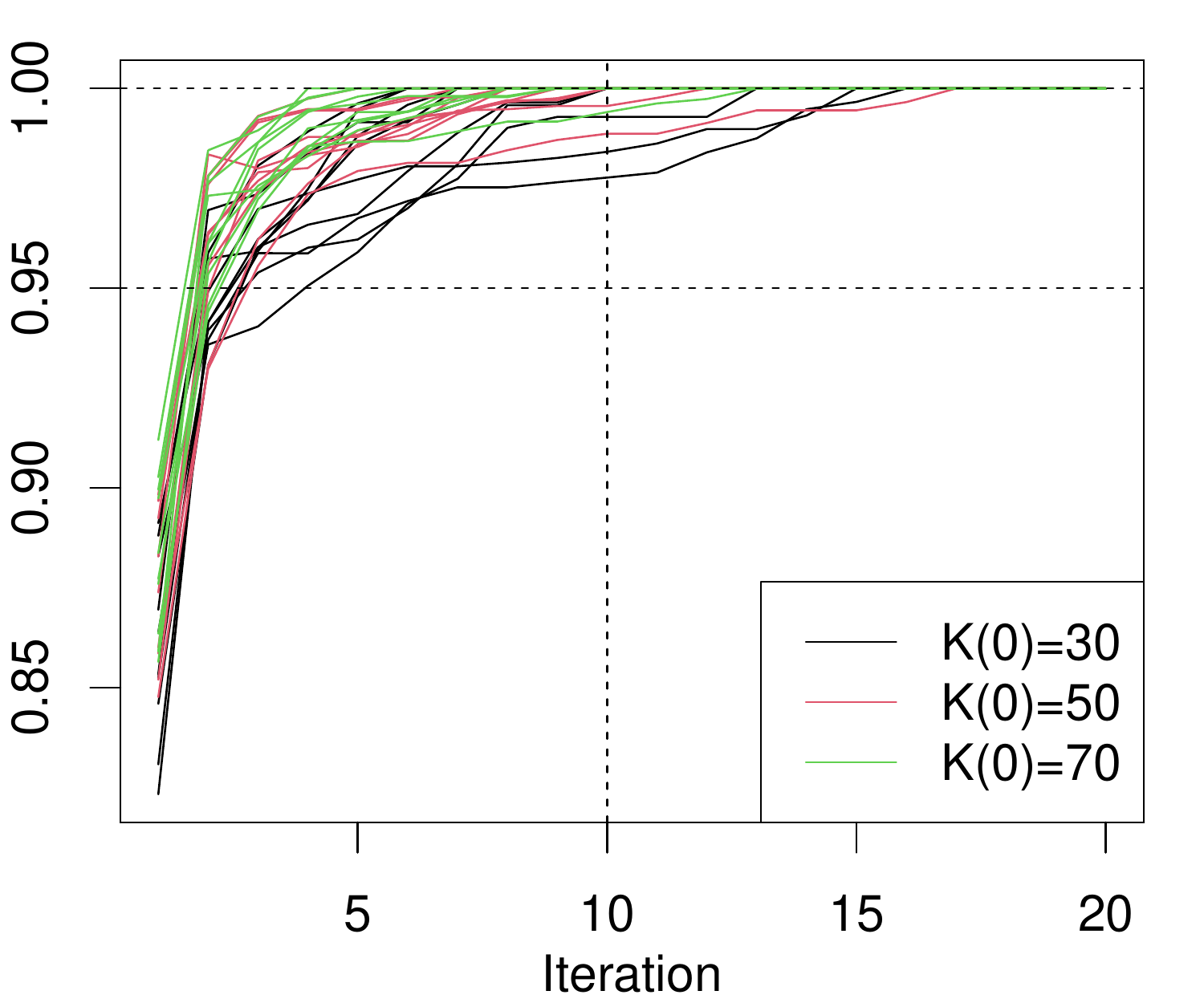}
\subcaption{Convergence of cluster assignments: $\eta_1(t)=\text{NMI}(\hat{\bs c}(t),\hat{\bs c}(\infty))$.}
\end{minipage}
\begin{minipage}{0.33\textwidth}
\centering
\includegraphics[width=\textwidth]{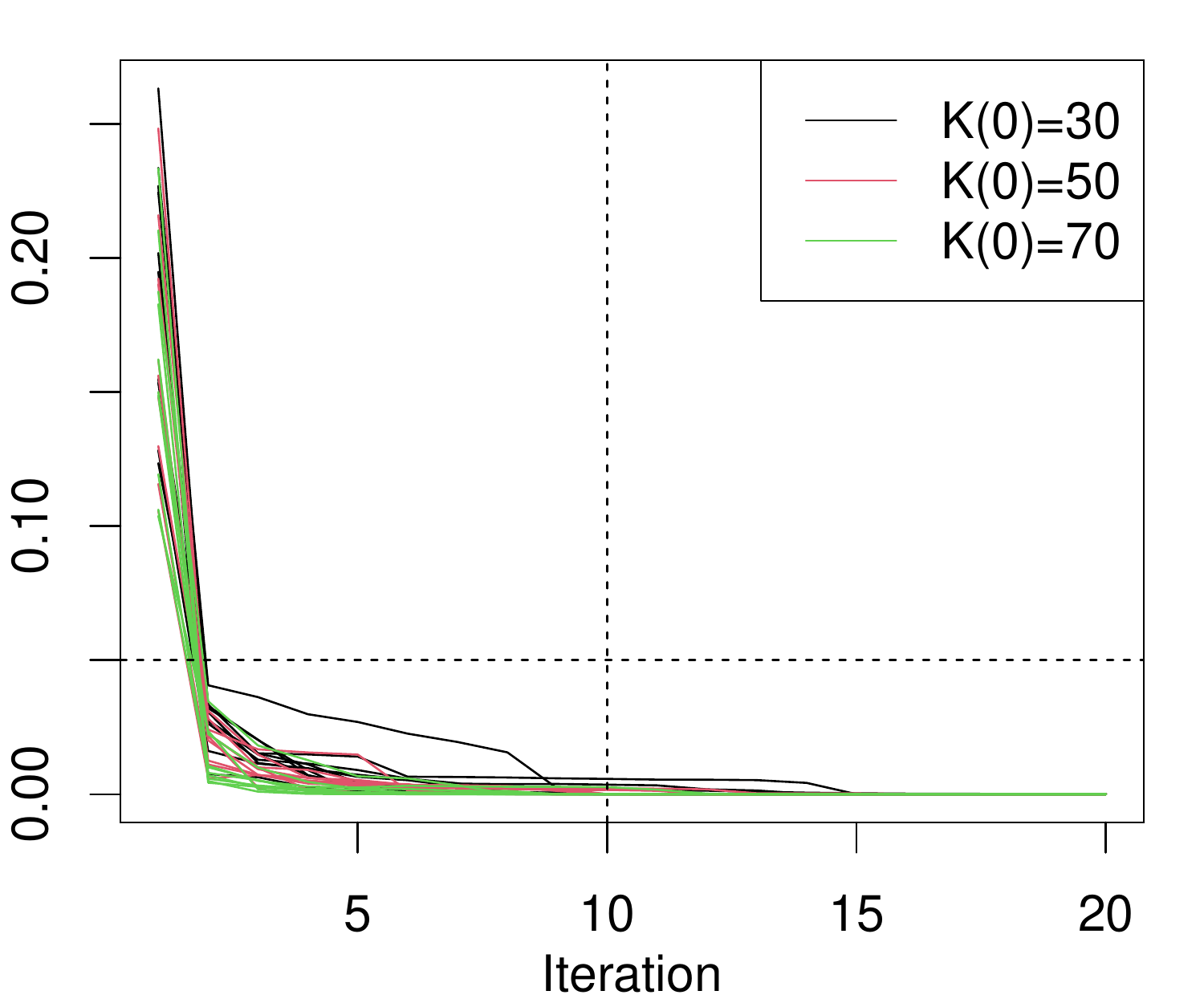}
\subcaption{Convergence of feature candidates: $\eta_2(t)=n^{-1}\sum_{i=1}^{n}\|\chi_i(t)-\chi_i(\infty))\|_2^2$.}
\end{minipage}
\begin{minipage}{0.33\textwidth}
\centering
\includegraphics[width=\textwidth]{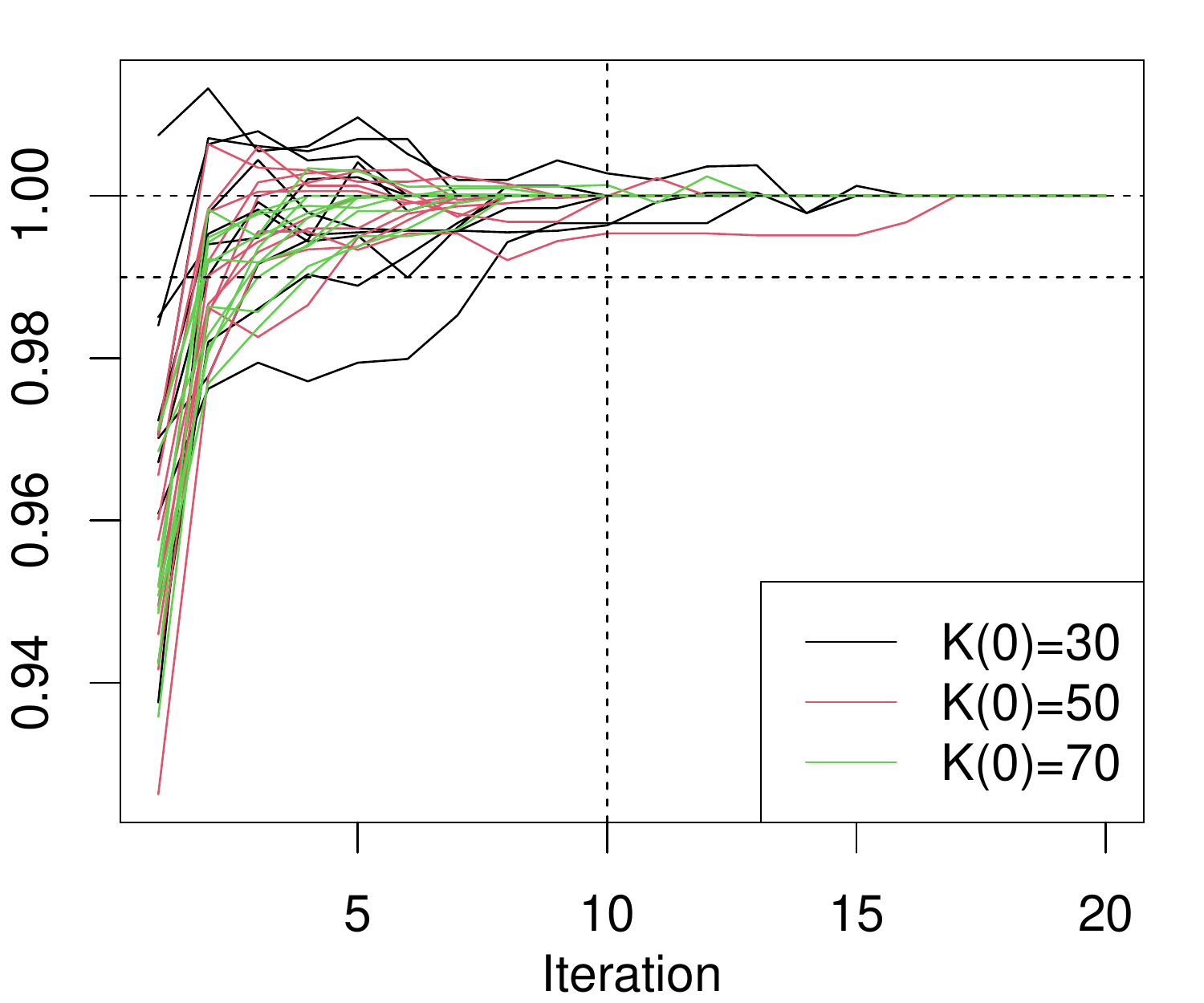}
\subcaption{Convergence of NMI: $\eta_3(t)=\text{NMI}(\hat{\bs c}(t),\bs c^*)/\text{NMI}(\hat{\bs c}(\infty),\bs c^*)$.}
\end{minipage}
\caption{Convergence, with $\lambda = 0.01$ and $K(0) \in \{30,50,70\}$.}
\label{fig:convergence}
\end{figure}

\paragraph{Results.} 
Although Experiments 1 and 2 show that almost $15$ iterations are required for the GOC algorithm to completely converge, both the cluster assignments and feature candidates almost converge within $\le 10$ iterations (specifically, the feature candidates almost converge within $\le 5$ iterations). NMI also nearly converges within $\le 5$ iterations, and only $5$--$10$ iterations are required for the GOC algorithm to obtain a sufficient performance during these experiments.

\subsubsection*{Experiment 4: Comparison to other clustering oracles.}

Table~\ref{table:comparison_to_other_oracles} computes the GOC algorithm and the corresponding baselines for other clustering oracles ($K$-medoids and the GMM using the \verb|ClusterR| and \verb|Mclust| packages). Note that the GMM~(\verb|Mclust|+BIC) selects the number of clusters using BIC (from $1,2,\ldots,K(t)$) in each clustering step.

\begin{table}[!ht]
\centering
\caption{Comparison to other clustering oracles with $\lambda=0.01$}
\label{table:comparison_to_other_oracles}
\subcaption{$K(0)=30$}
\begin{tabular}{llcccc} 
 \toprule 
 & & $K$-means & $K$-medoids & GMM~(\verb|ClusterR|) & GMM~(\verb|Mclust|+BIC) \\ 
 \midrule 
 \multirow{2}{*}{NMI} & GOC & $0.834\pm0.023$ & $0.848\pm0.028$ & $0.834\pm0.021$ & $0.802\pm0.028$ \\ 
 & Baseline  & $0.808\pm0.020$ & $0.816\pm0.028$ & $0.807\pm0.028$ & $0.781\pm0.025$ \\ 
 \cmidrule{2-6} 
\multirow{2}{*}{$F$-measure} & GOC & $0.641\pm0.026$ & $0.663\pm0.041$ & $0.635\pm0.033$ & $0.559\pm0.050$ \\ 
 & Baseline & $0.604\pm0.024$ & $0.626\pm0.038$ & $0.594\pm0.044$ & $0.528\pm0.046$ \\ 
 \cmidrule{2-6} 
$\#$clusters & GOC & $28.4\pm0.84$ & $28.7\pm0.95$ & $29.0\pm0.81$ & $22.8\pm3.39$ \\ 
 $\#$iterations & GOC & $15.8\pm3.74$ & $16.4\pm4.53$ & $18.8\pm3.33$ & $18.7\pm8.88$ \\ 
 \bottomrule 
 \end{tabular}
%%=======================================================================
%%=======================================================================
\subcaption{$K(0)=50$}
\begin{tabular}{llcccc} 
 \toprule 
 & & $K$-means & $K$-medoids & GMM~(\verb|ClusterR|) & GMM~(\verb|Mclust|+BIC) \\ 
 \midrule 
 \multirow{2}{*}{NMI} & GOC & $0.879\pm0.027$ & $0.879\pm0.024$ & $0.864\pm0.022$ & $0.807\pm0.035$ \\ 
 & Baseline  & $0.839\pm0.026$ & $0.841\pm0.026$ & $0.828\pm0.023$ & $0.786\pm0.032$ \\ 
 \cmidrule{2-6} 
\multirow{2}{*}{$F$-measure} & GOC & $0.752\pm0.046$ & $0.753\pm0.031$ & $0.718\pm0.039$ & $0.571\pm0.059$ \\ 
 & Baseline  & $0.685\pm0.048$ & $0.698\pm0.042$ & $0.654\pm0.04$ & $0.539\pm0.057$ \\ 
 \cmidrule{2-6} 
$\#$clusters & GOC & $46.4\pm0.97$ & $46.3\pm1.83$ & $48\pm1.05$ & $24\pm3.89$ \\ 
 $\#$iterations & GOC & $15\pm3.33$ & $16.2\pm2.94$ & $16.7\pm5.40$ & $18.3\pm9.11$ \\ 
 \bottomrule 
 \end{tabular}
%%=======================================================================
%%=======================================================================
\subcaption{$K(0)=70$}
\begin{tabular}{llcccc} 
 \toprule 
 & & $K$-means & $K$-medoids & GMM~(\verb|ClusterR|) & GMM~(\verb|Mclust|+BIC) \\ 
 \midrule 
 \multirow{2}{*}{NMI} & GOC & $0.878\pm0.022$ & $0.875\pm0.023$ & $0.868\pm0.023$ & $0.807\pm0.035$ \\ 
 & Baseline  & $0.837\pm0.021$ & $0.845\pm0.022$ & $0.839\pm0.022$ & $0.786\pm0.032$ \\ 
 \cmidrule{2-6} 
\multirow{2}{*}{$F$-measure} & GOC & $0.747\pm0.036$ & $0.743\pm0.037$ & $0.721\pm0.046$ & $0.571\pm0.059$ \\ 
 & Baseline  & $0.673\pm0.035$ & $0.692\pm0.039$ & $0.662\pm0.047$ & $0.539\pm0.057$ \\ 
 \cmidrule{2-6} 
$\#$clusters & GOC & $64.2\pm1.55$ & $63.9\pm2.03$ & $66.3\pm1.89$ & $24\pm3.89$ \\ 
 $\#$iterations & GOC & $12.2\pm2.10$ & $13.9\pm2.03$ & $16\pm5.48$ & $18.3\pm9.11$ \\ 
 \bottomrule 
 \end{tabular}
\end{table}

\paragraph{Results.} 
For all clustering oracles, the proposed GOC algorithm improves the clustering scores by simply applying the oracle to the representative vectors (baseline). 
In addition, $K$-means and $K$-medoids demonstrate almost the same performance for $K(0) = 50,70$, whereas $K$-means requires slightly fewer iterations to converge. 
We think that the vector $\hat{\mu}_k=(\ref{eq:mu_k})$ used to update the feature vectors is more compatible with $k$-means by minimizing the simple $\ell_2$-norm between the features and the cluster centers. 
GMM~(\verb|ClusterR|) demonstrates a similar performance as $K$-means and $K$-medoids, all of which detect almost the same number of clusters. 
GMM~(\verb|MClust|+BIC) detects fewer clusters than $K$-means, $K$-medoids, and GMM~(\verb|ClusterR|). 
GMM~(\verb|MClust|+BIC) also achieves lower scores, whereas GMM~(\verb|Mclust|+BIC) tends to detect the same cluster assignments and feature representative during the first iteration, regardless of $K(0)$ (because it selects the number of clusters using BIC). 
Note that the number of underlying true clusters is $K_*=50$, i.e., BIC applied to GMM~(\verb|Mclust|) underestimates the number of clusters. 
See Appendix~\ref{sec:convergence} for the convergence experiments conducted on $K$-medoids, GMM~(\verb|ClusterR|), and GMM~(\verb|Mclust|+BIC).

\paragraph{Notes on Experiment~4.} 
While we employed the EII model~($\Sigma_k=\sigma^2 I$) for $\verb|Mclust|$, we also conducted experiments on the \verb|Mclust| function with more general models, i.e., VII~($\Sigma_k=\sigma_k^2 I$ for $\sigma_k>0$) and VVV~(where $\Sigma_k \in \mathbb{R}^{d \times d}$ can be arbitrary): VVV generalizes VII, and VII generalizes EII. 
Under the setting $K(0)=50,\lambda=0.01$, NMI scores for the GOC algorithm using GMM~(\verb|Mclust|(VII)+BIC) and GMM~(\verb|Mclust|(VVV)+BIC) are $0.717 \pm 0.050$ and $0.662 \pm 0.055$, respectively, whereas the detected numbers of clusters are $13.0 \pm 3.97$ and $9.1 \pm 2.38$. Therefore, the scores are in the order of EII$ > $VII$ > $VVV, which is opposite the model degrees of freedom, which are in the order of EII$ < $VII$ < $VVV.

\subsection{Discussion 1: Similarity-based Clustering}
\label{subsec:discussion}

As another way to exploit the feature uncertainty, \citet{Kriegel2008density} and \citet{Jiang2013clsutering} define discrepancies between the probability densities $p_i,p_j$ of features $X_i,X_j$ and apply similarity-based clustering algorithms. 
However, we cannot employ this approach in our setting, as the explicit forms of the densities $p_i,p_j$ are hardly obtained due to the non-linear pre-processing. 
Therefore, as an alternative implementation, we measure the discrepancy between the feature uncertainty sets, and apply affinity propagation~\citep[AP;][]{frey2007clustering} which takes the similarity matrix as its input and outputs the estimated clusters. 

To compute the similarity matrix, for $1 \le i,j \le n$,
we employ the negative sign of the following three types of discrepancies between $\tilde{\setX}_i^{(m_i)},\tilde{\setX}_j^{(m_j)}$:
\begin{align*}
    s^{(1)}_{ij}
    &:=
    \frac{1}{m_i m_j}
    \sum_{\chi \in \tilde{\setX}_i^{(m_i)}}
    \sum_{\chi' \in \tilde{\setX}_j^{(m_j)}}
    \|\chi - \chi'\|_2, \\
%%===========================================
    s^{(2)}_{ij}
    &:=
    \min_{\chi \in \tilde{\setX}_i^{(m_i)}}
    \min_{\chi' \in \tilde{\setX}_j^{(m_j)}}
    \|\chi - \chi'\|_2, \\
%%============================================
    s^{(3)}_{ij}
    &:=
    \max\left\{
    \min_{\chi \in \tilde{\setX}_i^{(m_i)}}
    \max_{\chi' \in \tilde{\setX}_j^{(m_j)}}
    \|\chi - \chi'\|_2
    \,
    ,
    \,
    \min_{\chi' \in \tilde{\setX}_j^{(m_j)}}
    \max_{\chi \in \tilde{\setX}_i^{(m_i)}}
    \|\chi - \chi'\|_2
    \right\}.
\end{align*}
The last one is known as the Hausdorff distance.

The AP was implemented using the \verb|apcluster| package in \verb|R|. 
Therein, exemplar preferences are set to the sample quantile of the input dissimilarities with a threshold $q$; in addition, we employ $q = 0.5,0.7,0.9$. 
Table~\ref{table:comparison_to_AP} shows the NMI and $F$-measure for the baseline ($K$-means applied to the representative vectors~(\ref{eq:representative_vector})), and the AP equipped with negative signs of $S^{(1)}=(s^{(1)}_{ij}),S^{(2)}=(s^{(2)}_{ij}),S^{(3)}=(s^{(3)}_{ij})$.  
For the baseline and GOC algorithm, we set $K^{(0)} = 50$ and $\lambda = 0.01$, respectively.

\begin{table}[!ht]
\centering 
\caption{Comparison to affinity propagation}
\label{table:comparison_to_AP}
\begin{tabular}{llccc}
\toprule 
& & NMI & $F$-measure & $\#$clusters \\
\midrule
\multicolumn{2}{c}{GOC} & $0.879 \pm 0.027$ & $0.752 \pm 0.046$ & $46.4 \pm 0.97$ \\ 
\multicolumn{2}{c}{Baseline} & $0.839 \pm 0.026$ & $0.685 \pm 0.048$ & $50 \pm 0$ \\
\hline
\multirow{3}{*}{AP~($S^{(1)}$)} & $q=0.5$ & $0.801 \pm 0.022$ & $0.566 \pm 0.035$ & $23.5 \pm 2.80$ \\
& $q=0.7$ & $0.825 \pm 0.022$ & $0.623 \pm 0.035$ & $30.6 \pm 2.95$ \\
& $q=0.9$ & $0.848 \pm 0.024$ & $0.680 \pm 0.046$ & $55.0 \pm 4.83$ \\
\hline
\multirow{3}{*}{AP~($S^{(2)}$)} & $q=0.5$ & $0.751 \pm 0.048$ & $0.488 \pm 0.074$ &  $17.0 \pm 3.50$ \\
& $q=0.7$ & $0.793 \pm 0.029$ & $0.558 \pm 0.052$ & $20.7 \pm 3.20$ \\
& $q=0.9$ & $0.863 \pm 0.026$ & $0.717 \pm 0.048$ & $33.0 \pm 2.71$ \\
\hline
\multirow{3}{*}{AP~($S^{(3)}$)} & $q=0.5$ & $0.745 \pm 0.030$ & $0.500 \pm 0.047$ & $27.6 \pm 5.08$ \\
& $q=0.7$ & $0.772 \pm 0.025$ & $0.553 \pm 0.039$ &  $38.7 \pm 6.11$ \\
& $q=0.9$ & $0.816 \pm 0.016$ & $0.617 \pm 0.028$ & $74.9 \pm 7.48$ \\
\bottomrule
\end{tabular}
\end{table}

\paragraph{Results.} 
Although the GOC algorithm outperforms all AP clustering scores, the AP with $q = 0.9$ demonstrates a competitive performance (although the AP with a larger $q$ of greater than $0.9$ unfortunately tends to become unstable). 
The AP equipped with $S^{(2)}$ and $q = 0.9$ demonstrates good clustering scores. In fact, the discrepancy $s^{(2)}_{ij}$ does not satisfy the triangle inequality (namely, $s^{(2)}_{ij}$ is not the distance between $\tilde{\setX}_i^{(m_i)}$ and $\tilde{\setX}_j^{(m_j)}$), unlike $s^{(1)}_{ij}$ and $s^{(3)}_{ij}$. 
Therefore, for several specific situations, as an alternative to the proposed GOC algorithm, we admit the potential significance of distance-based clustering approaches endowed with some specific dissimilarities (that are not restricted to satisfying the definitions of distance ). 

Finally, we note that the complexity when computing the distances between pairs of ambiguity sets is large. 
The complexity is $O(n^2 m^2)$ when assuming that $m_1=m_2=\cdots=m_n=m$, whereas the GOC algorithm roughly requires $O(Tn^2)$ with $T$~iterations (in our numerical experiments, $T \approx 10$ and $m \approx 10^2$, whereby $m^2 \approx 10^4=10^3T$). 
Computing the similarities of the ambiguity sets is rather burdensome if $n$ and $m$ increase.

\subsection{Discussion 2: Convex Uncertainty Sets}
\label{subsec:discussion2_convex}

While we employ the empirical uncertainty set $\tilde{\setX}_i^{(m_i)}$ which is not restricted to be convex, we may consider an alternative convex set 
$\tilde{T}_i^{(m_i)}$ containing $\tilde{\setX}_i^{(m_i)}$: finding possible feature candidates over the set
\[
    \setA_n^{\dagger}
    :=
    \tilde{T}_1^{(m_1)} \times \tilde{T}_2^{(m_2)} \times \cdots \times \tilde{T}_n^{(m_n)}
\]
instead of $\setA_n$, is expected to be more efficiently computed by the existing optimization techniques related to convex sets. 
For computational reasons, \citet{Ngai2006efficient} considers a minimum box $B_i$ containing $\tilde{\setX}_i^{(m_i)}$ as the convex set $\tilde{T}_i^{(m_i)}$, and \citet{vo2016robust} assumes that the uncertainty set is box-shaped (i.e., convex).

Referring to \citet{vo2016robust}, we may employ difference-of-convex algorithm~\citep[DCA; see, e.g.,][]{Thi2005DC} to solve a specific form of GOC (particularly, GOC equipped with $K$-means) more efficiently. 
We think that this convex modification of GOC would be a future research worth considering, while we do not employ this convex set $\tilde{T}_i^{(m_i)}$ in this study by the following reasons: 
(i) to exploit the convex techniques, we need to heavily restrict the types of clustering oracle $\clustering$ (whereby the applicability of GOC would be much degraded, and the implementation would be mathematically difficult for users), and 
(ii) the convex set $\tilde{T}_i^{(m_i)}$ may contain large unnecessary regions in some situations (see Figure~\ref{fig:convex}).

\begin{figure}[!ht]
\centering
\begin{minipage}{0.45\textwidth}
\centering
    \includegraphics[width=0.5\textwidth]{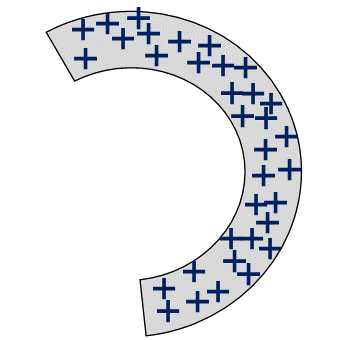}
    \subcaption{Underlying uncertianty set $\setX_i$ (colored in grey)}    
\end{minipage} 
\begin{minipage}{0.45\textwidth}
\centering
    \includegraphics[width=0.5\textwidth]{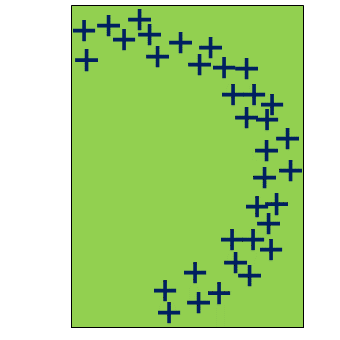}
    \subcaption{Smallest box $T_i$ containing $\tilde{\setX}_i^{(m_i)}$}    
\end{minipage} 
\caption{Compared to the underlying uncertainty set $\setX_i$ (where the empirical uncertainty set $\tilde{\setX}_i^{(m_i)}$ is shown by the ``+'' symbols), the smallest box $T_i$ containing $\tilde{\setX}_i^{(m_i)}$ includes a large unnecessary regions.}
\label{fig:convex}
\end{figure}

\section{Conclusion}
This study considered a clustering problem using user-specified uncertainty for the covariates. 
In particular, we considered a pre-processing that applies a non-linear transformation to the covariates to obtain features that are expected to capture the latent data structure. 
In addition, we proposed the GOC algorithm, which greedily finds better feature candidates over the uncertainty sets (of the pre-processed features). 
We applied the GOC algorithm to a synthetic orbital action dataset of sibling stars generated through our numerical simulation, for which the proposed algorithm improved the clustering scores. 
We also provided realistic datasets and source codes to reproduce the experimental results in \url{https://github.com/oknakfm/GOC}. 

Finally, we describe below the limitations and possible extensions of this study. 

\paragraph{Limitations.} We must specify $K(0)$ and $\lambda \ge 0$ when applying the GOC algorithm. 
Specifying the number of clusters is a common and historical problem in computational statistics (for instance, see \citet{Thorndike53whobelongs} for the elbow method, and inexhaustible discussions have been developed for several decades). 
Although we may simply apply the existing approaches to each clustering step, the number of clusters detected depends on the feature candidate at a particular step under our problem setting, and it remains unclear whether the conventional approaches are still effective. 
The BIC used in the \verb|Mclust| implementation of the GMM underestimates the number of clusters, as shown through Experiment~4 described in Section~\ref{subsec:experimental_results}. Because clustering is an unsupervised problem, regarding the selection of hyperparameter $\lambda \ge 0$, we cannot employ standard statistical approaches such as a cross-validation. 
Although our experiment results show the adequacy of using a small $\lambda$ (even a $\lambda$ of 0 is effective with our datasets), we have yet to sufficiently confirm this.

\paragraph{Possible extensions of this research.} 
A possible extension of this study would be to accelerate the speed of the GOC algorithm. 
Because most of the clustering algorithms considered in this study~(i.e., $K$-means, $K$-medoids, and GMM) use iterative algorithms, we can terminate the iterations in each clustering oracle before convergence is reached. 
Namely, we can reduce the number of unnecessary iterations (within each clustering step) and focus more on the convergence in the sense of the overall GOC. 
Another possible extension of this research is to incorporate hierarchical clustering into the GOC algorithm. Because the iterations of the algorithm easily break down the hierarchical structure found during the clustering step, some modifications of the algorithm are needed to obtain the hierarchical structure of the GOC output.

\section*{Acknowledgement}
AO was supported by JSPS KAKENHI (Grant No. JP21K17718) and JST CREST (Grant No. JPMJCR21N3). 
KH was also supported by JSPS KAKENHI (Grant Nos. JP21K13965 and JP21H00053).
We would like to thank Keisuke Yano for helpful discussion

\appendix

\section{Detailed Descriptions of Synthetic Dataset}
\label{subsec:detailed_descriptions_of_synthetic_dataset}

\subsection{General Description of the Simulation}

In galactic astronomy, it is believed that the Milky Way was formed through the merging of smaller systems, such as dwarf galaxies. 
In the numerical experiments described in Section \ref{subsec:synthetic_dataset_generation}, we generated mock data by simulating the formation process of the Milky Way. 
To simplify this case, we assume that $K_* = 50$ dwarf galaxies merge with the Milky Way and are instantaneously disrupted at time $\tau = 0$. 
Each dwarf galaxy contains 30,000 sibling stars. 
When a dwarf galaxy is disrupted, 
sibling stars begin moving independently. 
After $\tau = 0$, and until the current epoch ($\tau = 10 \times 10^{9}$ years), 
the motions of these sibling stars are treated as test particles (i.e., particles with zero mass) moving within the gravitational potential of the Milky Way. 
For each dwarf galaxy, the positions and velocities of the sibling stars at $\tau = 0$ (i.e., the initial conditions) slightly differ from each other. 
The small difference in the initial conditions evolves over cosmic time, and the positions and velocities of the sibling stars are completely different from each other in the current epoch, although they originate from the same dwarf galaxy. 

\subsection{Visualization of the Simulation}
\label{subsec:visualization_of_the_simulation}

To provide an intuitive understanding of the simulation,  Figure~\ref{fig:KH_visualize_simulation} shows a subset of sibling stars in two dwarf galaxies A and B that merge with the Milky Way at $\tau = 0$. 
At $\tau = 0$, the sibling stars in each dwarf galaxy have identical positions and slightly different velocities, making these two groups clearly distinguishable in terms of their positions and velocities. 
At $\tau = 3 \times 10^{9}$ years, the positions and velocities of the sibling stars exhibit a wider distribution, and the two groups of stars are marginally distinguishable in terms of their positions and velocities. 
At the current epoch, $\tau = 10 \times 10^{9}$ years, the positions and velocities of the sibling stars show a mixed distribution. 
At this point, it is difficult to separate two groups of stars from each other in terms of their positions and velocities.

From these three snapshots, it is evident that finding sibling stars within a six-dimensional position and velocity space becomes more difficult as the system evolves over time. 
Importantly, this difficulty is unrelated to the accuracy of the data. 
Even if we have a perfect measurement of the positions and velocities of the stars, finding sibling stars is a difficult task if we use the raw data of position and velocity.

However, the case appears to be simpler if we look at the system in a three-dimensional phase space spanned by the orbital action. 
(As a reminder of the readers, the orbital action is a three-dimensional conserved quantity, which is a function of position and velocity; and it describes the stellar orbital properties.)
%(i.e., a three-dimensional conserved quantity, as a function of the position and velocity, describing the stellar orbital properties.)
Because the orbital action is conserved for each star, the distribution of sibling stars in the orbital action space is also conserved over time, as shown in the rightmost panels in Figure~\ref{fig:KH_visualize_simulation}. 
Therefore, using the action distribution instead of the position and velocity distributions is an indispensable strategy for identifying sibling stars. 
As shown in Example \ref{ex:stars}, in the presence of observational uncertainties in the stellar positions and velocities, it is difficult to find sibling stars in the action space, 
which motivated us to introduce the GOC algorithm as a new type of clustering approach. 

\begin{figure}[!ht]
\centering
\includegraphics[width=0.9\textwidth]{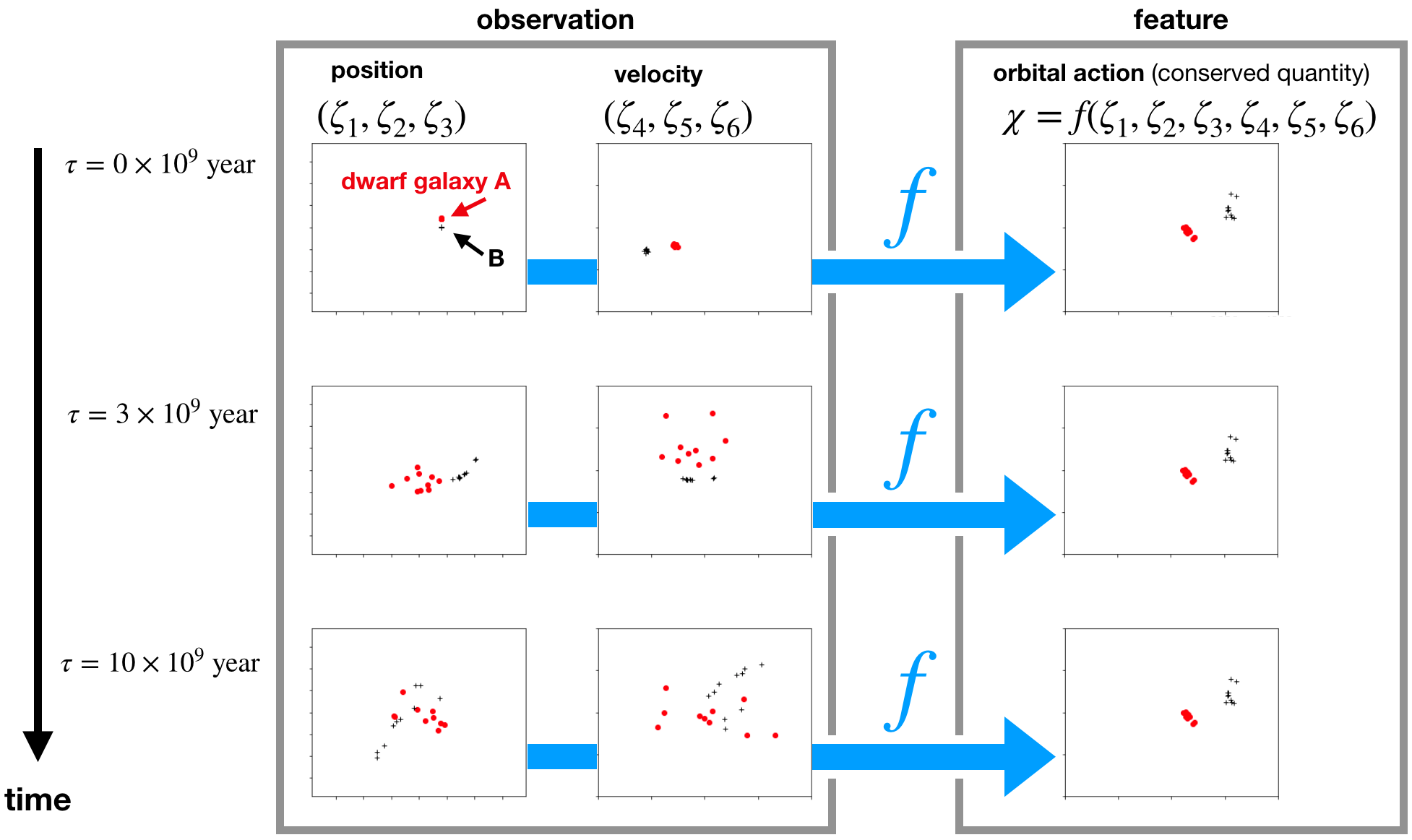}\\
\caption{Disruption of two dwarf galaxies in the Milky Way. 
Even in the absence of observational uncertainty, 
finding sibling stars from a six-dimensional position and the velocity space is difficult at the current epoch ($\tau = 10\times10^{9}$ years).
By contrast, finding sibling stars in a three-dimensional orbital action space is easier 
because the orbital actions are conserved over time. 
Note that, for clarity, this figure only shows two-dimensional projections of the three-dimensional position, velocity, and orbital action. 
}
\label{fig:KH_visualize_simulation}
\end{figure}

\subsection{Detailed Implementations of the Simulation}

To generate mock data, we first randomly generated $K_* = 50$ centroids with positions $\vec{x}^\mathrm{centroid}_k$ and velocities $\vec{v}^\mathrm{centroid}_k$ ($k=1,\cdots,K_*$) using a realistic distribution function model of the Milky Way, similar to that described in \cite{Hattori2021MNRAS.508.5468H}. 
The $k$th centroid corresponds to the position and velocity of the $k$th dwarf galaxy at $\tau = 0$. 
For the $k$th dwarf galaxy, we generated $N_\mathrm{sibling}=30000$ positions and velocities, $\vec{x}_{ks}(\tau=0)$ and $\vec{v}_{ks}(\tau=0)$ ($s=1, \cdots, N_\mathrm{sibling}$), 
such that 
$\vec{x}_{ks}(\tau=0) = \vec{x}^\mathrm{centroid}_k$ and 
$\vec{v}_{ks}(\tau=0) \sim \normaldist(\vec{v}^\mathrm{centroid}_k, \Sigma)$ with 
$\Sigma = (5 \;\mathrm{km\;s^{-1}})^2 I$. 
Here, $I \in \mathbb{R}^{3 \times 3}$ denotes the identity matrix. 
These positions and velocities correspond to the initial conditions of the sibling stars at $\tau = 0$. 
From these initial conditions, we integrated the orbits of $K_{*} N_\mathrm{sibling}$ stars for $10\times 10^{9}$ years (which is approximately the age of the universe) under a widely used gravitational potential model of the Milky Way described in \cite{McMillan2017} and derived the current-day positions and velocities. 
At this point, $N_\mathrm{sibling}$ stars originating from the same dwarf galaxy are no longer located close to each other (see the bottom-left panel in Figure~\ref{fig:KH_visualize_simulation}). 
To mimic the observations, for the $k$th group, we randomly select $n_k$ stars that are close to the current position of the Sun. 
(Note that stars that are too far away from the Sun are too faint to be observed.) 
The assumed position and velocity of the Sun is the same as those in \cite{Doke2022arXiv220315481D}. 
We chose $n_k = 1 + \{(k-1) \text{ mod } 10\}$, 
where ``$x \text{ mod } a$'' denotes the residual of the division ($x$ divided by $a$). 
With this, 
we have $n = 275$ stars in total, 
such that we have 1 member star for $k = 1,11,21,31,41$; 
we have 2 member stars for $k = 2,12,22,32,42$; and so on.

For completeness, in the following, 
we briefly mention how we converted the simulated data 
into the uncertainty set used by the GOC algorithm. 
(See Section \ref{subsec:synthetic_dataset_generation} for a full description.)
First, we converted 
the simulated stellar positions and velocities of $n = 275$ stars 
into observable quantities, 
as illustrated schematically in Figure~\ref{fig:KH1}(\subref{fig:zeta123456}). 
Note that 
the stellar positions and velocities in the simulation are 
{\it true} quantities that are unavailable in reality. 
To mimic the actual observation, 
we add a random error to the observable quantities, 
which are then used to construct the uncertainty set.

By following the same procedure, we run 10 independent simulations.  
Each simulation is used to construct a dataset.

\section{Convergence of GOC Using $K$-Medoids and GMM}
\label{sec:convergence}

Regarding the comparison of the convergence of the $K$-means clustering shown through Experiment~3 described in Section~\ref{subsec:experimental_results}, Figures~\ref{fig:convergence_k-medoids}--\ref{fig:convergence_GMM_EII} show the convergence of $K$-medoids, GMM~(\verb|ClusterR|), and GMM~(\verb|Mclust|+BIC), respectively, all of which demonstrated the same tendencies.

\begin{figure}[!ht]
\centering
\begin{minipage}{0.33\textwidth}
\centering
\includegraphics[width=\textwidth]{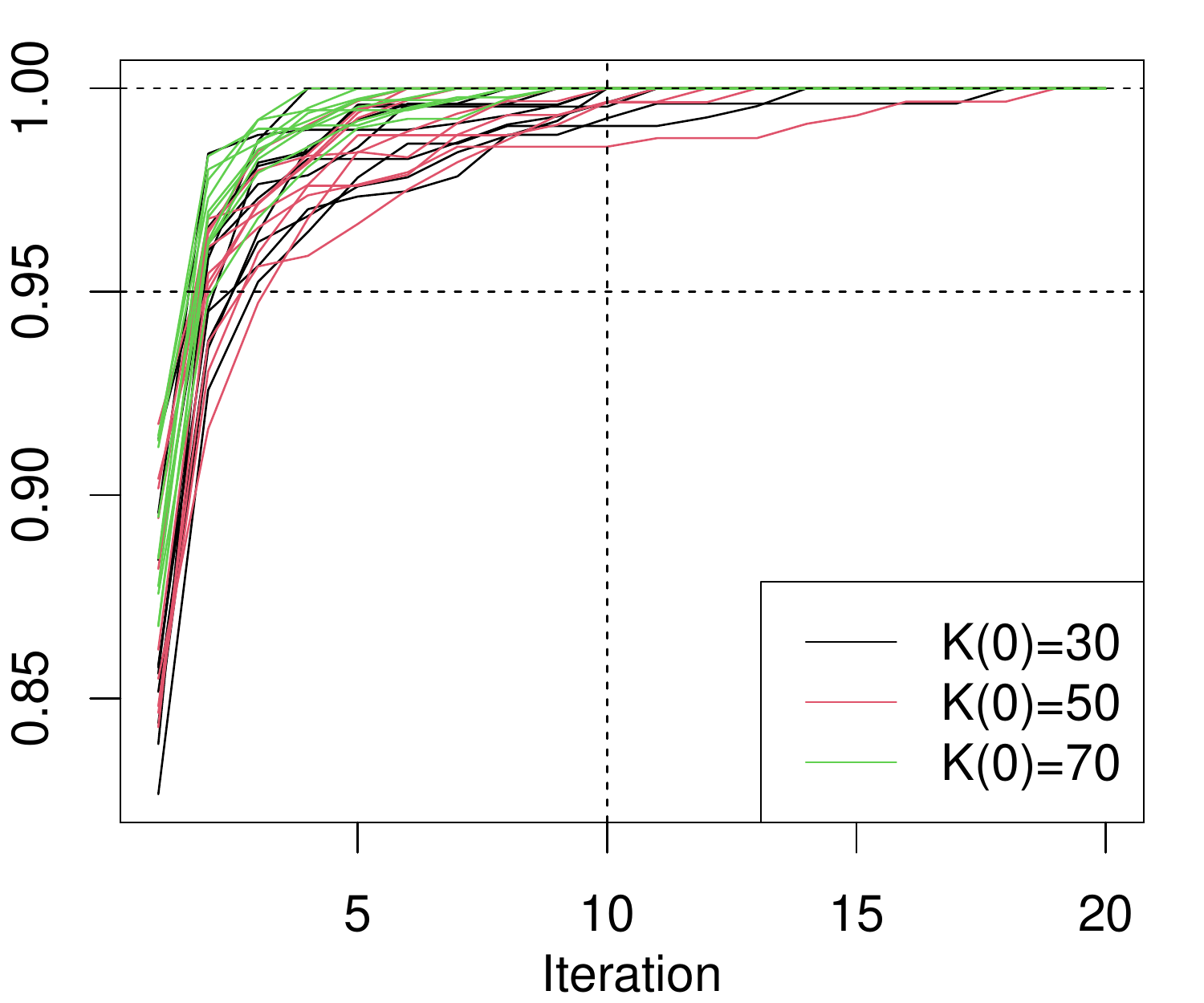}
\subcaption{$\text{NMI}(\hat{\bs c}(t),\hat{\bs c}(\infty))$.}
\end{minipage}
\begin{minipage}{0.33\textwidth}
\centering
\includegraphics[width=\textwidth]{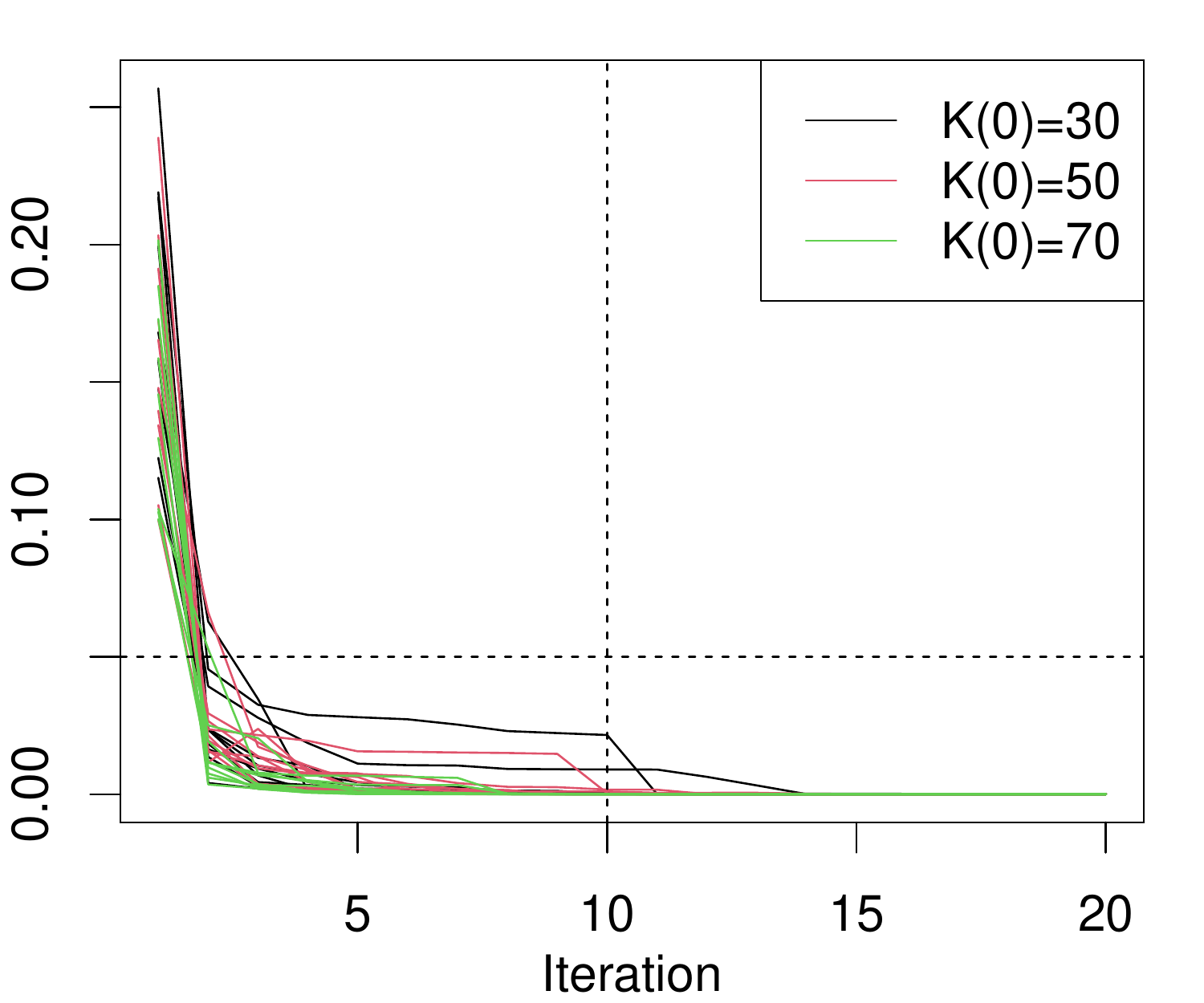}
\subcaption{$n^{-1}\sum_{i=1}^{n}\|\chi_i(t)-\chi_i(\infty))\|_2^2$.}
\end{minipage}
\begin{minipage}{0.33\textwidth}
\centering
\includegraphics[width=\textwidth]{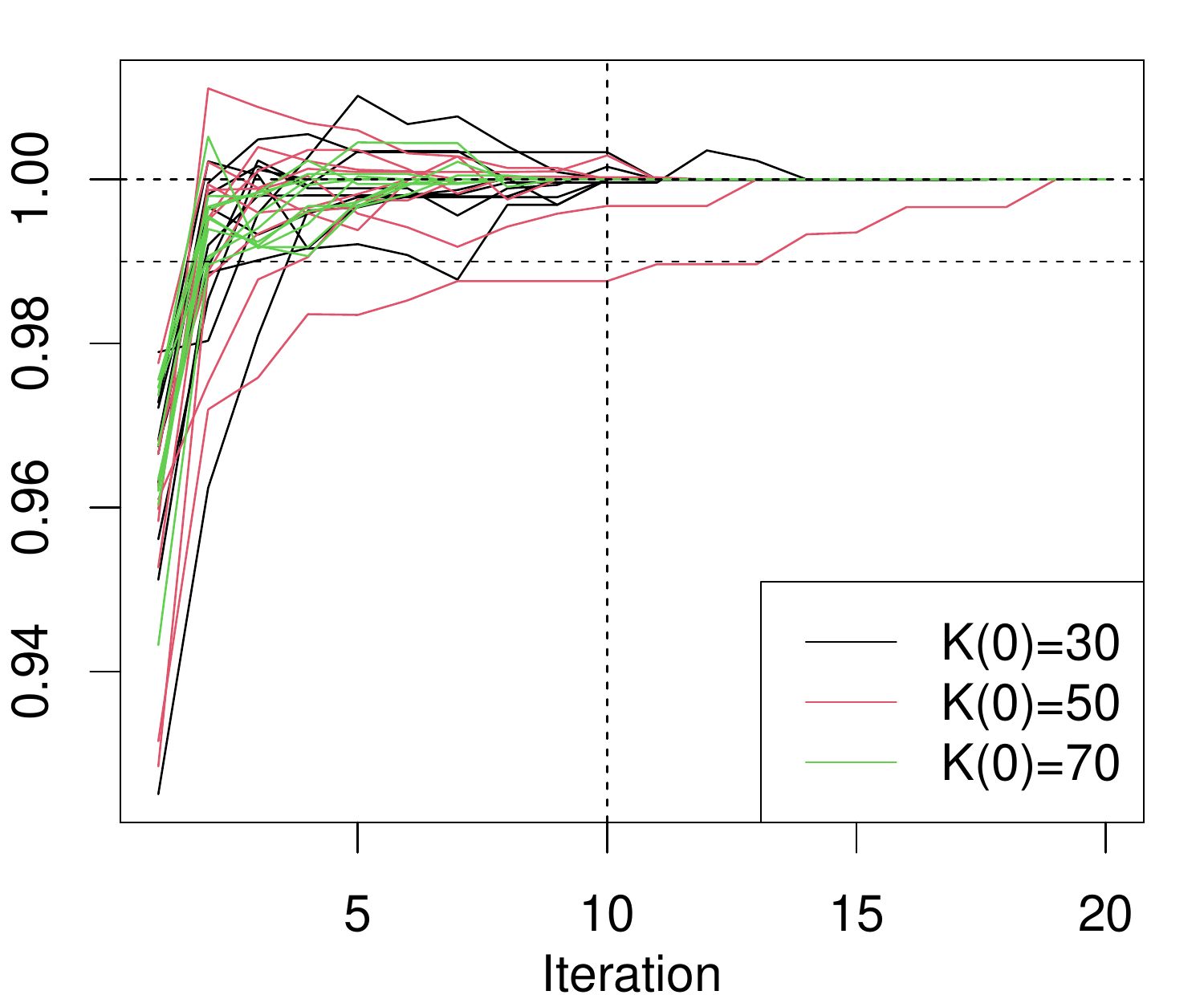}
\subcaption{$\text{NMI}(\hat{\bs c}(t),\bs c^*)/\text{NMI}(\hat{\bs c}(\infty),\bs c^*)$.}
\end{minipage}
\caption{Convergence of $K$-medoids.}
\label{fig:convergence_k-medoids}
%%====================================================
\begin{minipage}{0.33\textwidth}
\centering
\includegraphics[width=\textwidth]{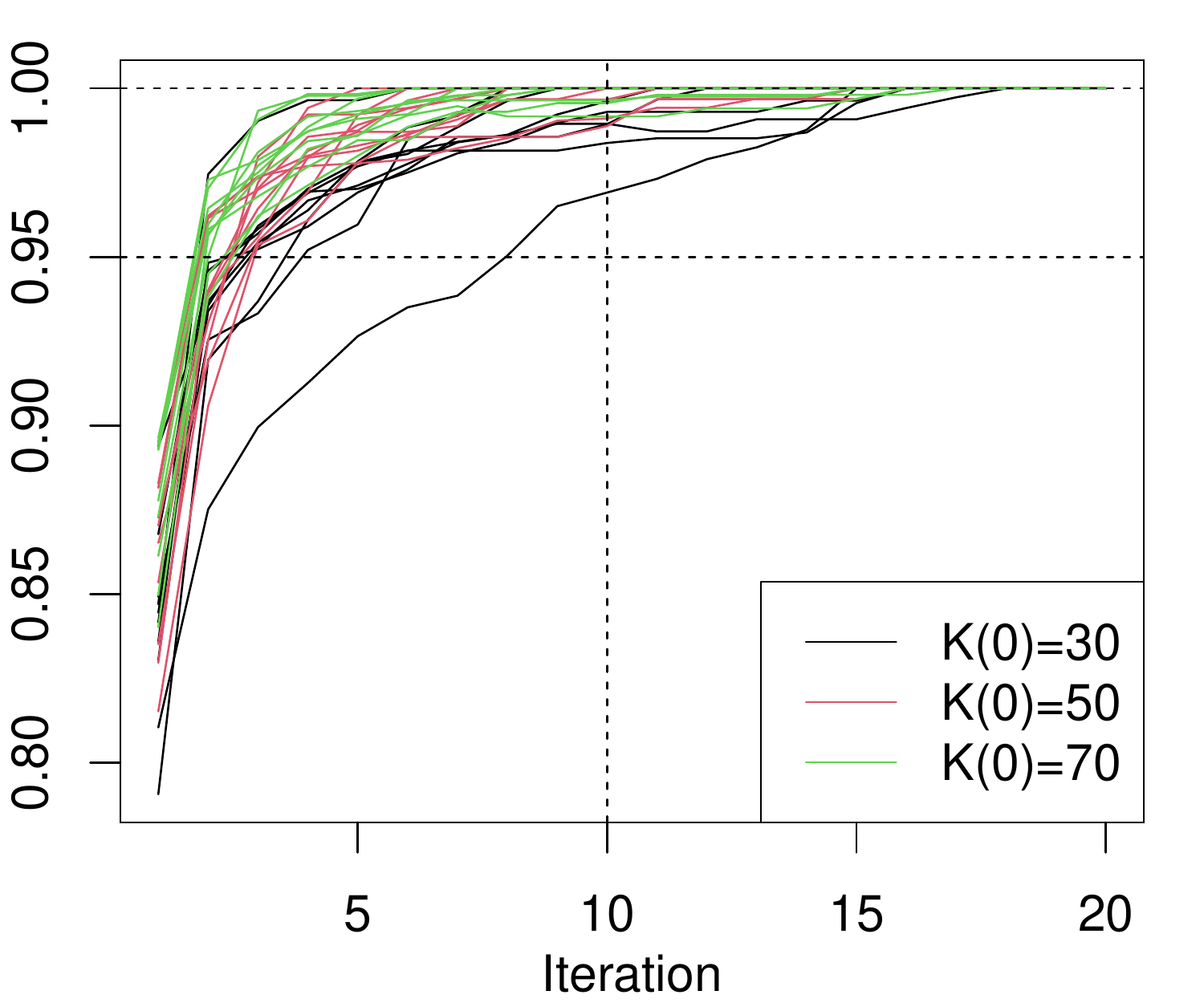}
\subcaption{$\text{NMI}(\hat{\bs c}(t),\hat{\bs c}(\infty))$.}
\end{minipage}
\begin{minipage}{0.33\textwidth}
\centering
\includegraphics[width=\textwidth]{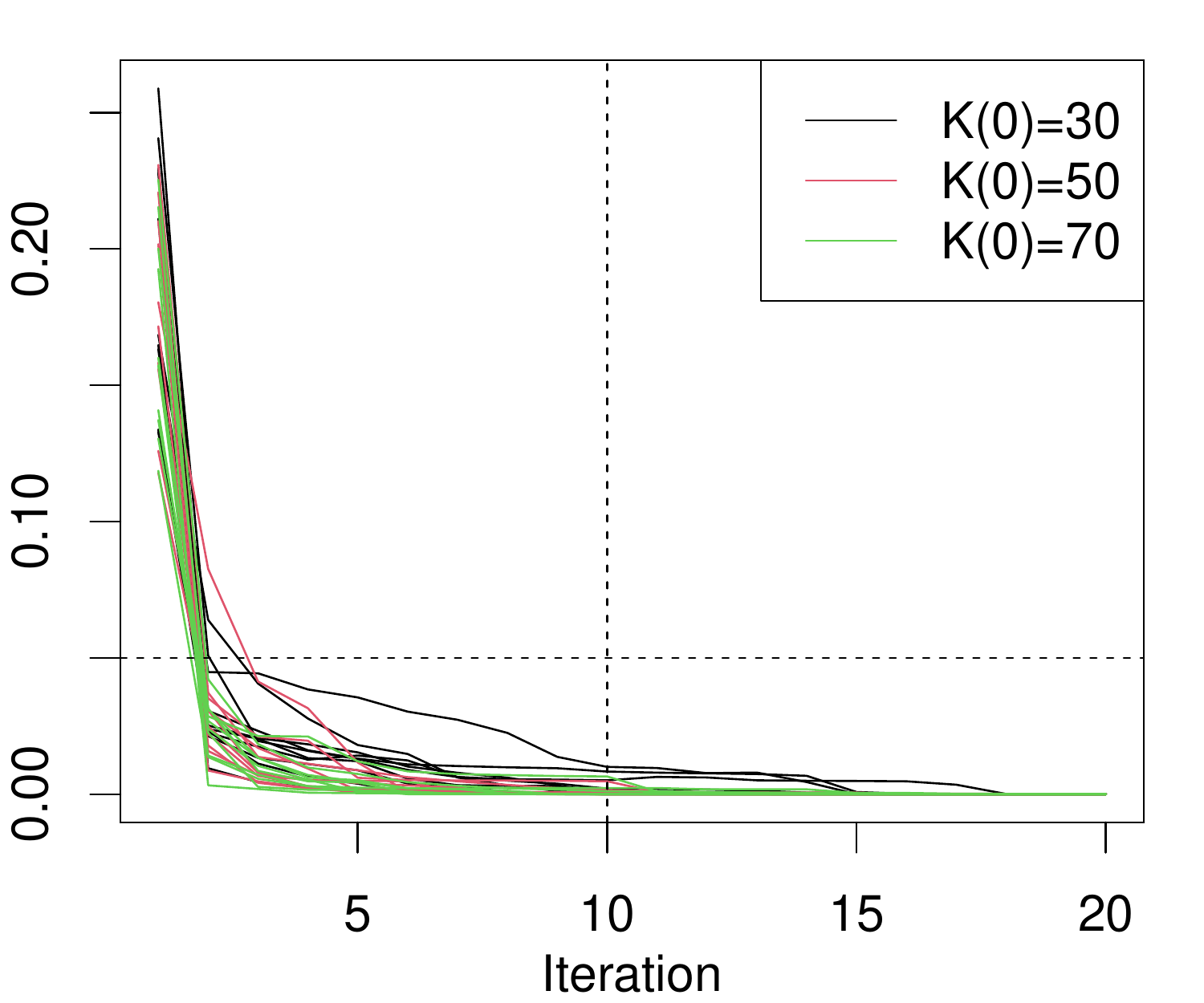}
\subcaption{$n^{-1}\sum_{i=1}^{n}\|\chi_i(t)-\chi_i(\infty))\|_2^2$.}
\end{minipage}
\begin{minipage}{0.33\textwidth}
\centering
\includegraphics[width=\textwidth]{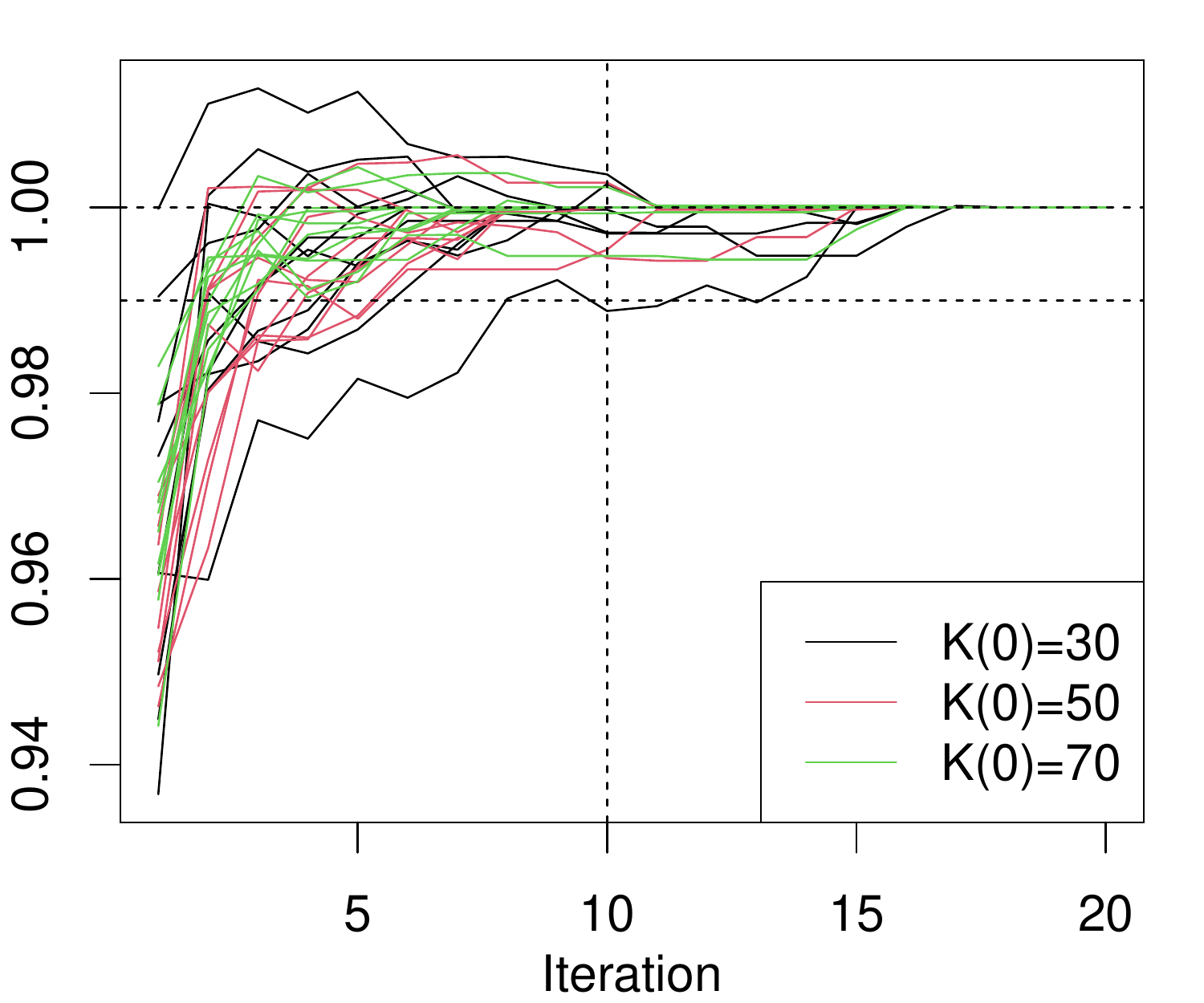}
\subcaption{$\text{NMI}(\hat{\bs c}(t),\bs c^*)/\text{NMI}(\hat{\bs c}(\infty),\bs c^*)$.}
\end{minipage}
\cprotect\caption{Convergence of GMM~(\verb|ClusterR|).}
\label{fig:convergence_GMM_ClusterR}
%%====================================================
\begin{minipage}{0.33\textwidth}
\centering
\includegraphics[width=\textwidth]{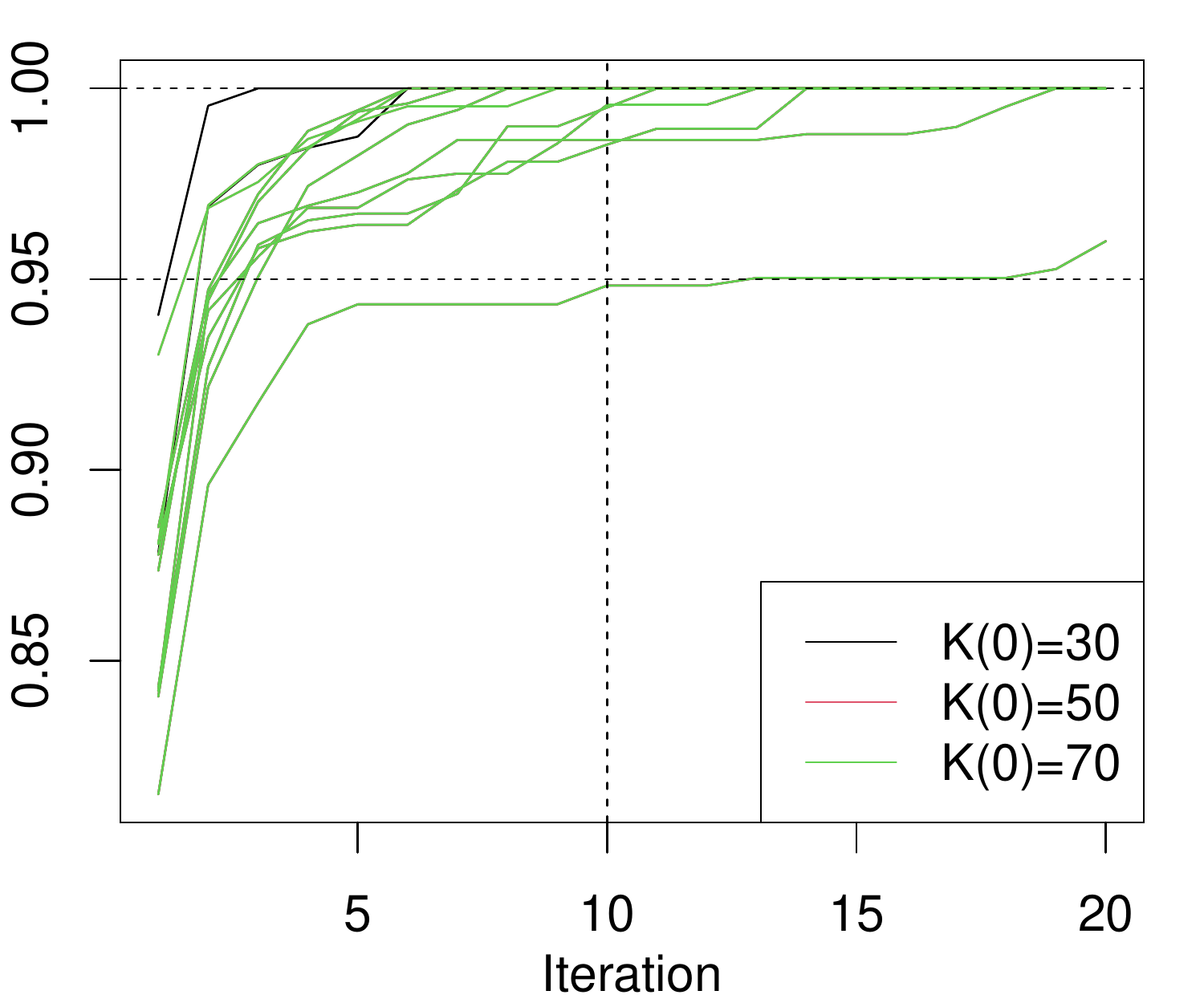}
\subcaption{$\text{NMI}(\hat{\bs c}(t),\hat{\bs c}(\infty))$.}
\end{minipage}
\begin{minipage}{0.33\textwidth}
\centering
\includegraphics[width=\textwidth]{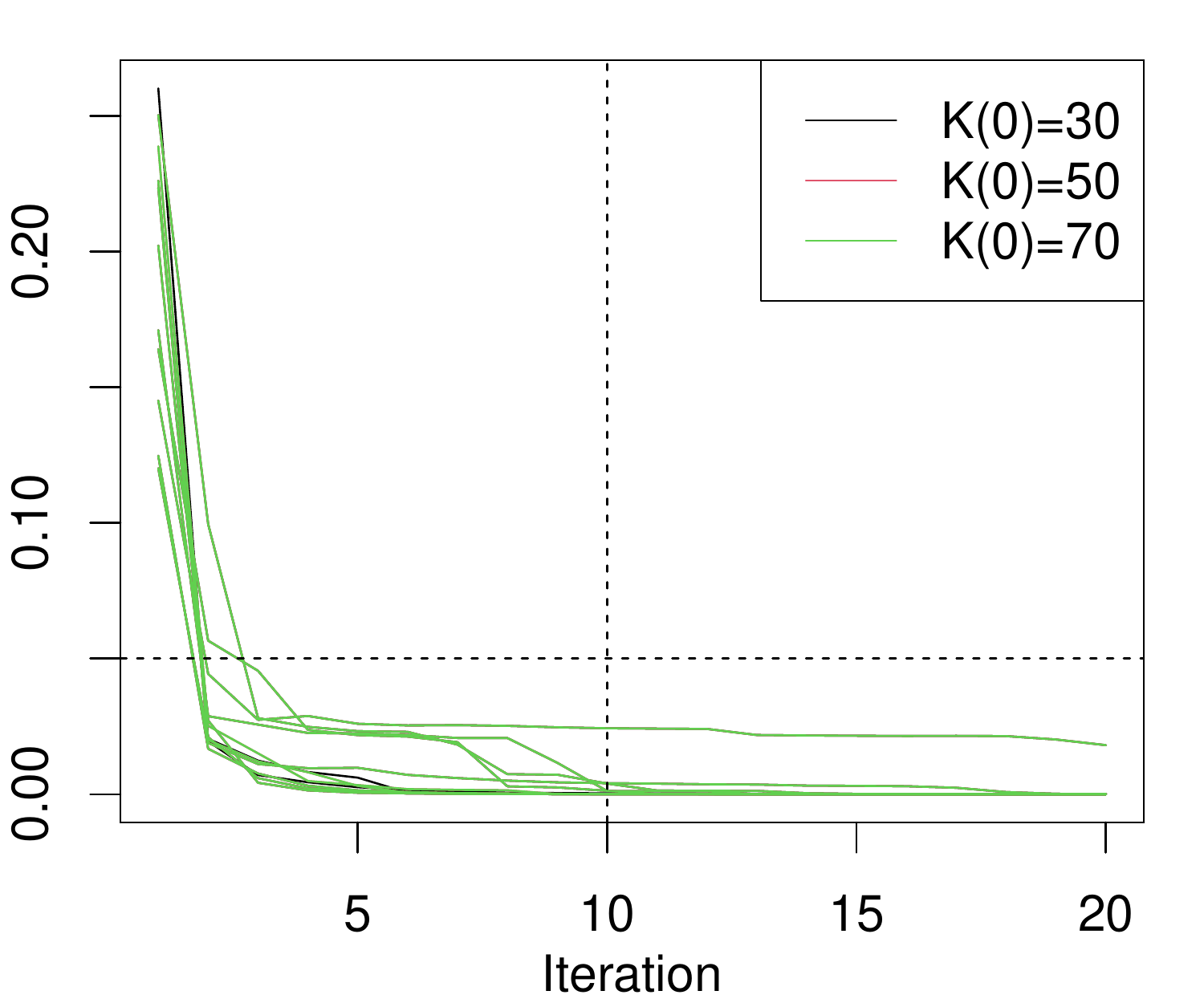}
\subcaption{$n^{-1}\sum_{i=1}^{n}\|\chi_i(t)-\chi_i(\infty))\|_2^2$.}
\end{minipage}
\begin{minipage}{0.33\textwidth}
\centering
\includegraphics[width=\textwidth]{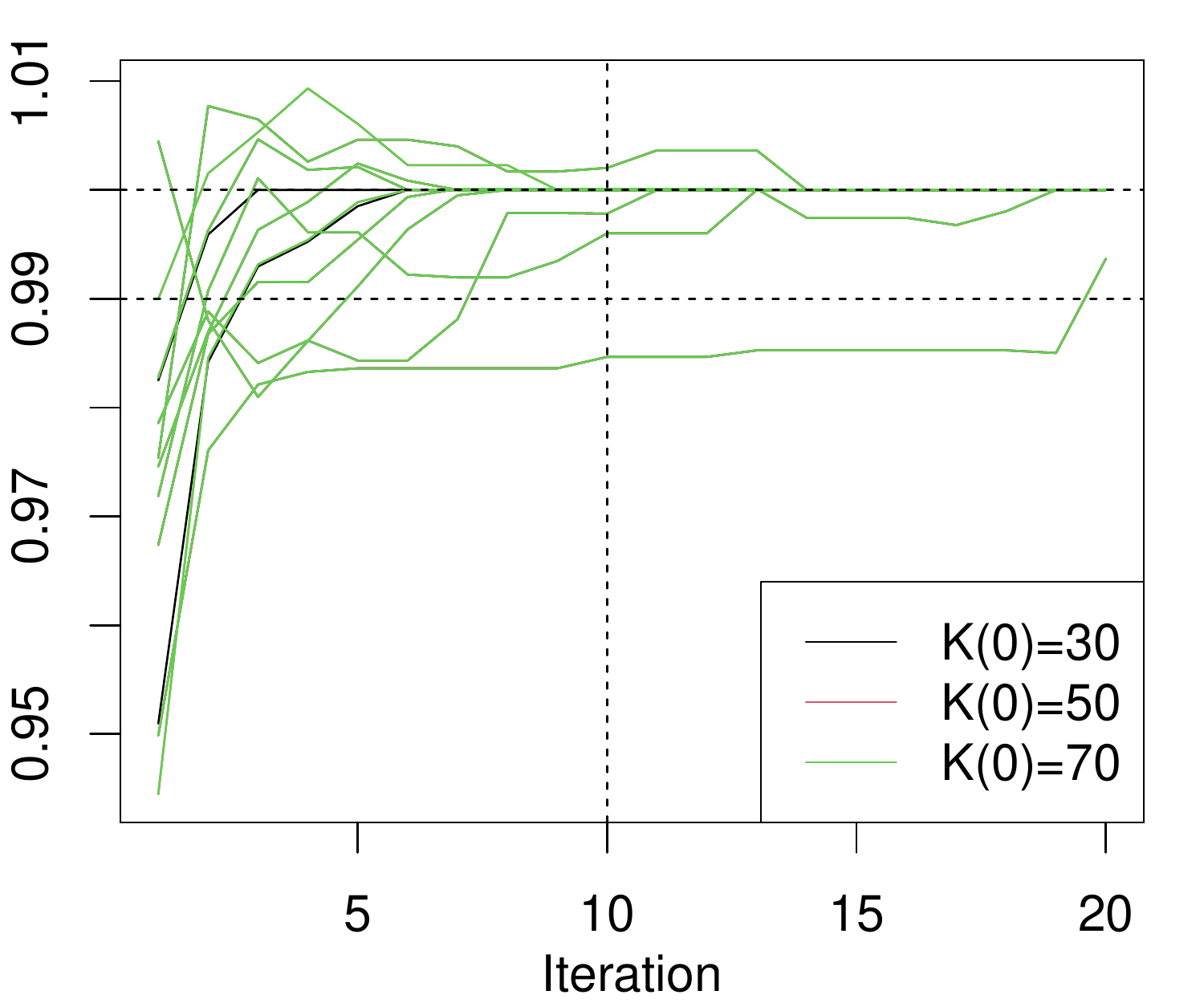}
\subcaption{$\text{NMI}(\hat{\bs c}(t),\bs c^*)/\text{NMI}(\hat{\bs c}(\infty),\bs c^*)$.}
\end{minipage}
\cprotect\caption{Convergence of GMM~(\verb|Mclust|+BIC).}
\label{fig:convergence_GMM_EII}
\end{figure}

% \bibliographystyle{apalike}
% \bibliography{oc}

\end{document}